\documentclass[creativecommons, withnoncommercial]{eptcs}
\providecommand{\event}{ICE 2020} 

\usepackage{soul}
\usepackage{underscore}           
\usepackage{hyperref}
\usepackage{amsmath}
\usepackage{amssymb}
\usepackage{mathtools}
\usepackage{stmaryrd}
\usepackage{graphicx}
\usepackage{subfigure}
\usepackage{multicol}
\usepackage[usenames,dvipsnames,svgnames,table]{xcolor} 
\usepackage{tikz}
\usepackage{xstring}
\usepackage{graphics,framed}
\usepackage[capitalise]{cleveref}
\crefformat{enumi}{condition~#2#1#3}
\crefname{fact}{Fact}{Facts}
\Crefname{fact}{Fact}{Facts}
\crefformat{fact}{Fact~#2#1#3}
\usepackage{xargs}
\usepackage{titlecaps}
\usepackage{listings}
\usepackage{bussproofs}
\usepackage{etoolbox}
\usepackage{extarrows}
\usepackage{lineno}

\usepackage{xspace}

\usepackage[inline]{enumitem}

\newif\iffinal
\finalfalse

\usepackage[textsize=tiny%
\iffinal
,disable%
\fi
]{todonotes}

\iffinal
\else
\fi

\usepackage{amsthm}
\newcounter{mycounter}
\newtheorem{definition}{Definition}
\newtheorem{theorem}[mycounter]{Theorem}
\newtheorem{corollary}[mycounter]{Corollary}

\newtheorem{lemma}[mycounter]{Lemma}
\newtheorem{example}[mycounter]{Example}
\newtheorem{remark}[mycounter]{Remark}

\newif\ifemi\emifalse
\newif\ifcr\crfalse
\newif\ifcr\crtrue

\newcommand{\tnxbehapi}[1][partly]{ Research {#1} supported by the EU
  H2020 RISE programme under the Marie Skłodowska-Curie grant
  agreement No 778233
}

\newcommand{\tnxitmatters}[1][partially]{
  Work #1
  funded by MIUR project PRIN 2017FTXR7S \emph{IT MATTERS} (Methods
  and Tools for Trustworthy Smart Systems)
}

\DeclareGraphicsExtensions{%
    .png,.PNG,%
    .pdf,.PDF,%
    .jpg,.mps,.jpeg,.jbig2,.jb2,.JPG,.JPEG,.JBIG2,.JB2}

\usepackage{bm}
\usepackage{fixme}

\usepackage[normalem]{ulem} 

\usepackage{xifthen}        
\newcommand{\ifempty}[3]{%
  \ifthenelse{\isempty{#1}}{#2}{#3}%
}

\newcommand{\mkfun}[4][\colorFun]{
  \newcommand{#2}[1][#4]{
    {#1\textsf{#3}}
    \ifempty{##1}{}{
      \ {##1}}
  }
}

\newcommand{\mkuop}[4][\colorFun]{
  \newcommand{#2}[1][#4]{
    {#1\textsf{#3}}
    \ifempty{##1}{}{
      \, {##1}}
  }
}

\newcommand{\hidden}[1]{}

\newcommand{\cf}[2]{
  \fontsize{#1}{#1}{\selectfont{#2}}
}
\ifemi
\usepackage{showlabels}

\newcommand{\emi}[2]{
  \marginpar{\fcolorbox{red}{shadecolor}{\cf{#1}{{#2}}}}
}
\newcommand{\emic}[2]{\par
  \fcolorbox{red}{shadecolor}{\parbox{\linewidth}{ 
      \color{gray}
      \begin{description}
      \item[{\color{blue} #2}]{\sf #1}
      \end{description}}}
}
\else
\newcommand{\emi}[2]{}
\newcommand{\emic}[2]{}{}
\fi

\newcommand{\sst}{\;\big|\;}
\newcommand{\qst}{\;\colon\;} 

\newcommand{\conf}[1]{\ensuremath{\langle {#1} \rangle}}

\renewcommand{\vec}[1]{\overline{#1}}
\newcommand{\bnfdef}{\ ::=\ }
\newcommand{\bnfmid}{\;\ \big|\ \;}

\newcommand{\qqand}[1][and]{\qquad\text{#1}\qquad}
\newcommand{\qand}[1][and]{\quad\text{#1}\quad}

{\bfseries}{\rmfamily}
{\bfseries}{\rmfamily}

\newcommand{\squo}[1]{\lq {#1}\rq}
\newcommand{\quo}[1]{\lq\lq {#1}\rq\rq}
\def\finex{{\unskip\nobreak\hfil
\penalty50\hskip1em\null\nobreak\hfil$\diamond$
\parfillskip=0pt\finalhyphendemerits=0\endgraf}}

\definecolor{shadecolor}{rgb}{1,0.99,0.9}
\definecolor{bg}{rgb}{0.95,0.95,0.95}

\input{ggmacros}

\newcommand{\conflict}{\#}
\newcommand{\Set}[1]{\{ #1 \}}
\newcommand{\Then}{\Longrightarrow}
\newcommand{\Sem}[1]{[\![  #1 ]\!]}

\mkfun{\aSort}{sort}{\al}

\newcommand{\Tensor}{\otimes}
\newcommand{\Or}{\;\vee\;}
\newcommand{\msgSet}{{\cal M}}
\newcommand{\PartSet}{{\cal P}}
\newcommand{\ChanSet}{{\cal C}}
\newcommand{\inj}{\iota}
\newcommand{\Domain}[1]{{\cal D}( #1)}
\newcommand{\Proj}[2]{#1  \upharpoonright  #2}
\newcommand{\MaxC}[1]{M( #1 )}
\renewcommand{\MaxC}[1]{{\cal C}_{\max}(#1)}

\newcommand{\co}[1]{\textit{co}({#1})}

\newcommandx{\atype}[2][1=\minP,2=\maxP,usedefault=@]{\conf{#1,#2}}
\newcommand{\tj}[4]{#1 \vdash #2 : \atype[{#3}][{#4}]}
\newcommand{\minP}{\phi}
\newcommand{\maxP}{\Lambda}

\newcommand{\CxtPi}{\rm \Pi}

\newcommand{\compch}[3]{#1\bowtie_{#3}#2}

\newcommand{\arule}[3]{\frac{\textstyle\rule[-.8ex]{0cm}{3ex}#1}%
{\textstyle\rule[.2ex]{0cm}{3ex}#2}{\mbox{\scriptsize {\sc #3}}}}

\title{Towards Refinable Choreographies\ifcr\thanks{
	 \tnxbehapi.
	 \tnxitmatters.
	 Research partially supported by the UBACyT projects
	 20020170100544BA and 20020170100086BA, and by the PIP project
	 11220130100148CO.
	 COST Action: EUTypes CA15123; and local funds of the University of
	 Turin: Ricerca locale Linea A (BERS-RILO-17-03) - Fondazioni
	 logiche della computazione, Ricerca locale Linea A
	 (PAOL-RILO-18-01) - Fondazioni logiche della computazione, Ricerca
	 locale 2019 Linea A (DE-U-RILO-19-01) - Logica della computazione.
	 \newline The authors thank the anonymous reviewers for their
	 comments and the interesting discussions on the forum of ICE20.
  }
  \fi
}

\author{
  Ugo de'Liguoro \institute{Universit\`a di Torino, Italy} \email{}
  \and
  Hern\'an Melgratti \institute{
	 \begin{tabular}[c]{c}
	  ICC \\  Universidad de Buenos Aires \\ Conicet, Argentina
	 \end{tabular}
}
  \and
  Emilio Tuosto \institute{
	 \begin{tabular}[c]{c}
School of Informatics, Leicester, UK \\ GSSI, L'Aquila, Italy
	 \end{tabular}
}
}
\def\authorrunning{U. de' Liguoro, H. Melgratti, E. Tuosto}

\def\titlerunning{Towards Refinable Choreographies}

\newdimen\proofrulebreadth \proofrulebreadth=.05em
\newdimen\proofdotseparation \proofdotseparation=1.25ex
\newdimen\proofrulebaseline \proofrulebaseline=2ex
\newcount\proofdotnumber \proofdotnumber=3
\let\then\relax
\def\hfi{\hskip0pt plus.0001fil}
\mathchardef\squigto="3A3B
\newif\ifinsideprooftree\insideprooftreefalse
\newif\ifonleftofproofrule\onleftofproofrulefalse
\newif\ifproofdots\proofdotsfalse
\newif\ifdoubleproof\doubleprooffalse
\let\wereinproofbit\relax
\newdimen\shortenproofleft
\newdimen\shortenproofright
\newdimen\proofbelowshift
\newbox\proofabove
\newbox\proofbelow
\newbox\proofrulename
\def\shiftproofbelow{\let\next\relax\afterassignment\setshiftproofbelow\dimen0 }
\def\shiftproofbelowneg{\def\next{\multiply\dimen0 by-1 }%
\afterassignment\setshiftproofbelow\dimen0 }
\def\setshiftproofbelow{\next\proofbelowshift=\dimen0 }
\def\setproofrulebreadth{\proofrulebreadth}

\def\prooftree{
\ifnum  \lastpenalty=1
\then   \unpenalty
\else   \onleftofproofrulefalse
\fi
\ifonleftofproofrule
\else   \ifinsideprooftree
        \then   \hskip.5em plus1fil
        \fi
\fi
\bgroup
\setbox\proofbelow=\hbox{}\setbox\proofrulename=\hbox{}%
\let\justifies\proofover\let\leadsto\proofoverdots\let\Justifies\proofoverdbl
\let\using\proofusing\let\[\prooftree
\ifinsideprooftree\let\]\endprooftree\fi
\proofdotsfalse\doubleprooffalse
\let\thickness\setproofrulebreadth
\let\shiftright\shiftproofbelow \let\shift\shiftproofbelow
\let\shiftleft\shiftproofbelowneg
\let\ifwasinsideprooftree\ifinsideprooftree
\insideprooftreetrue
\setbox\proofabove=\hbox\bgroup$\displaystyle 
\let\wereinproofbit\prooftree
\shortenproofleft=0pt \shortenproofright=0pt \proofbelowshift=0pt
\onleftofproofruletrue\penalty1
}

\def\eproofbit{
\ifx    \wereinproofbit\prooftree
\then   \ifcase \lastpenalty
        \then   \shortenproofright=0pt  
        \or     \unpenalty\hfil         
        \or     \unpenalty\unskip       
        \else   \shortenproofright=0pt  
        \fi
\fi
\global\dimen0=\shortenproofleft
\global\dimen1=\shortenproofright
\global\dimen2=\proofrulebreadth
\global\dimen3=\proofbelowshift
\global\dimen4=\proofdotseparation
\global\count255=\proofdotnumber
$\egroup  
\shortenproofleft=\dimen0
\shortenproofright=\dimen1
\proofrulebreadth=\dimen2
\proofbelowshift=\dimen3
\proofdotseparation=\dimen4
\proofdotnumber=\count255
}

\def\proofover{
\eproofbit 
\setbox\proofbelow=\hbox\bgroup 
\let\wereinproofbit\proofover
$\displaystyle
}%
\def\proofoverdbl{
\eproofbit 
\doubleprooftrue
\setbox\proofbelow=\hbox\bgroup 
\let\wereinproofbit\proofoverdbl
$\displaystyle
}%
\def\proofoverdots{
\eproofbit 
\proofdotstrue
\setbox\proofbelow=\hbox\bgroup 
\let\wereinproofbit\proofoverdots
$\displaystyle
}%
\def\proofusing{
\eproofbit 
\setbox\proofrulename=\hbox\bgroup 
\let\wereinproofbit\proofusing
\kern0.3em$
}

\def\endprooftree{
\eproofbit 
  \dimen5 =0pt
\dimen0=\wd\proofabove \advance\dimen0-\shortenproofleft
\advance\dimen0-\shortenproofright
\dimen1=.5\dimen0 \advance\dimen1-.5\wd\proofbelow
\dimen4=\dimen1
\advance\dimen1\proofbelowshift \advance\dimen4-\proofbelowshift
\ifdim  \dimen1<0pt
\then   \advance\shortenproofleft\dimen1
        \advance\dimen0-\dimen1
        \dimen1=0pt
        \ifdim  \shortenproofleft<0pt
        \then   \setbox\proofabove=\hbox{%
                        \kern-\shortenproofleft\unhbox\proofabove}%
                \shortenproofleft=0pt
        \fi
\fi
\ifdim  \dimen4<0pt
\then   \advance\shortenproofright\dimen4
        \advance\dimen0-\dimen4
        \dimen4=0pt
\fi
\ifdim  \shortenproofright<\wd\proofrulename
\then   \shortenproofright=\wd\proofrulename
\fi
\dimen2=\shortenproofleft \advance\dimen2 by\dimen1
\dimen3=\shortenproofright\advance\dimen3 by\dimen4
\ifproofdots
\then
        \dimen6=\shortenproofleft \advance\dimen6 .5\dimen0
        \setbox1=\vbox to\proofdotseparation{\vss\hbox{$\cdot$}\vss}%
        \setbox0=\hbox{%
                \advance\dimen6-.5\wd1
                \kern\dimen6
                $\vcenter to\proofdotnumber\proofdotseparation
                        {\leaders\box1\vfill}$%
                \unhbox\proofrulename}%
\else   \dimen6=\fontdimen22\the\textfont2 
        \dimen7=\dimen6
        \advance\dimen6by.5\proofrulebreadth
        \advance\dimen7by-.5\proofrulebreadth
        \setbox0=\hbox{%
                \kern\shortenproofleft
                \ifdoubleproof
                \then   \hbox to\dimen0{%
                        $\mathsurround0pt\mathord=\mkern-6mu%
                        \cleaders\hbox{$\mkern-2mu=\mkern-2mu$}\hfill
                        \mkern-6mu\mathord=$}%
                \else   \vrule height\dimen6 depth-\dimen7 width\dimen0
                \fi
                \unhbox\proofrulename}%
        \ht0=\dimen6 \dp0=-\dimen7
\fi
\let\doll\relax
\ifwasinsideprooftree
\then   \let\VBOX\vbox
\else   \ifmmode\else$\let\doll=$\fi
        \let\VBOX\vcenter
\fi
\VBOX   {\baselineskip\proofrulebaseline \lineskip.2ex
        \expandafter\lineskiplimit\ifproofdots0ex\else-0.6ex\fi
        \hbox   spread\dimen5   {\hfi\unhbox\proofabove\hfi}%
        \hbox{\box0}%
        \hbox   {\kern\dimen2 \box\proofbelow}}\doll%
\global\dimen2=\dimen2
\global\dimen3=\dimen3
\egroup 
\ifonleftofproofrule
\then   \shortenproofleft=\dimen2
\fi
\shortenproofright=\dimen3
\onleftofproofrulefalse
\ifinsideprooftree
\then   \hskip.5em plus 1fil \penalty2
\fi
}

\begin{document}

\ifemi\setcounter{tocdepth}{2}\listoffixmes\fi

\maketitle              
\begin{abstract}
  We investigate refinement in the context of choreographies.
  We introduce refinable global choreographies allowing for the
  underspecification of protocols, whose interactions can be refined
  into actual protocols.
  Arbitrary refinements may spoil \emph{well-formedness}, that is the
  sufficient conditions that guarantee a protocol to be
  implementable.
  We introduce a typing discipline that enforces well-formedness of
  typed choreographies.
  Then we unveil the relation among refinable choregraphies and their
  admissible refinements in terms of an axiom scheme.
\end{abstract}


\section{Introduction}\label{sec:intro}
%
The advent of structured programming~\cite{dij76} is probably behind
the widespread use of refinement methods in computer science.
Refinement is paramount in many formal methods, in software
engineering, and in verification, because the possibility of
structuring a system into simpler components is crucial to tackle the
complexity of a system.

In this paper we investigate the refinement of choreographies of
message-passing systems.
In this domain, a choreography specifies the coordination of
distributed components (aka \emph{participants} or \emph{roles}) by
disciplining the exchange of messages.
Following W3C~\cite{w3c:cho}, we envisage a choreography as a contract
consisting of a \emph{global view} that can be used as a blueprint for
defining each participant.
A global view is basically an \emph{application-level} protocol
realised through the coordination of the resulting \emph{local views},
the specifications of participants.
This description is the ground for the so-called top-down engineering
represented by the following diagram:
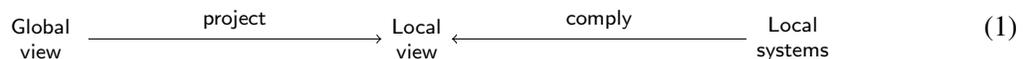
\begin{equation}
  \label{eq:chor:dia}
  \begin{tikzpicture}[node distance=1cm and 1cm, every text node part/.style={align=center}]
    \node at (0,0) (g) [text width = 1cm] {Global view};
    \node at (5,0) (l) {Local\\view};
    \node at (10,0) (c) {Local\\systems};
    \draw[->] (g) -- (l) node [midway,above]{project};
    \draw[->] (c) -- (l) node [midway,above]{comply};
  \end{tikzpicture}
\end{equation}
where the \squo{projection} operation produces local views from the
global ones and the operation \squo{comply} verifies that the
behaviour of each participant adheres to the one of the corresponding
local view.

Choreographic approaches are appealing because, unlike orchestration,
they do not require an explicit coordinator (see~\cite{bdft16} for a
deeper discussion).
Moreover, global views allow developers to work independently on
different components.

Despite the main advantages discussed above, choreographic approaches
suffer a main drawback: the lack of support for modular development.
This shortcoming is present in standards such as BPMN or in workflow
patterns and languages~\cite{boe12} and it has been more recently
flagged also for choreographic programming~\cite{cmv18}.

We propose a choreographic model of message-passing applications based
on point-to-point communication equipped with a simple refinement
mechanism.
Let us illustrate this through some simple examples.
This gives us the opportunity to informally use \emph{global
  choreographies}~\cite{gt18,gt16} (g-choreographies, for short), the
formalisation of global views adopted here for the technical
development of the paper.

Consider the g-choreography
\begin{align}\label{ex:tie}
  \gcho[][{\refgint[c][md][s]}][{\gseq[][{\refgint[c][req][s]}][{\refgint[s][done][c]}]}]
\end{align}
where a client \p[c] either sends some meta-data $\msg[md]$ or a
request $\msg[req]$ to a server \p[s].
In the former case the protocol terminates, while in the latter the
server is supposed to send back a response $\msg[done]$ to \p[c].
The dashed arrows above represent \emph{refinable interactions},
that is interactions that can be replaced so to refine the
application-level
protocol.
For instance, to allow \p[s] to send \p[c] some statistical
information in the second branch of \eqref{ex:tie} we can refine
$\refgint[s][done][c]$ with
$\gseq[][{\gint[][s][stats][c]}][{\gint[][s][done][c]}]$ and obtain
\begin{align}\label{ex:tieref}
  \gcho[][{\refgint[c][md][s]}][{\gseq[][{\refgint[c][req][s]}][{\gseq[][{\gint[][s][stats][c]}][{\gint[][s][done][c]}]}]}]
\end{align}
where the interactions with the solid arrow are now \quo{ground},
namely they cannot be further refined.

This is our simple refinement mechanism: replace a refinable
interaction with a more complex (refinable) protocol.
A key goal here is to provide a mechanism of refinement without
spoiling \emph{well-formedness} conditions.
Basically, well-formedness conditions ensure that the application-level
protocol modelled by the global view is faithfully executed by the
participants that comply with the projected local views.
Let us again explain this with an example.
Suppose we refine \eqref{ex:tie} by replacing each refinable
interaction with its ground version but for $\refgint[c][md][s]$,
which is replaced by
$\gseq[][{\gint[][c][md][b]}][{\gint[][b][md][s]}]$, where \q\ is a
brokerage service mediating the exchange of $\msg[md]$.
We obtain
\begin{align}\label{eq:err}
  \gcho[][{\gseq[][{\gint[][c][md][b]}][{\gint[][b][md][s]}]}][{\gseq[][{\gint[][c][req][s]}][{\gint[][s][done][c]}]}]
\end{align}
The g-choreography above is not well-formed because the broker \p[b]
is oblivious of the second branch.
Namely, \p[b] will be stuck waiting for message $\msg[md]$ should
\p[C]\ opt for the second branch of the choice.

\paragraph{Contributions \& Structure}
We introduce a simple mechanism for refining global views of choreographies.
Firstly, we equip an existing formal language expressing \emph{global
  choreographies} (g-choreographies, for short) with a semantics based
on event structures (surveyed in \cref{sec:bkg}) and identifying a
typing discipline (\cref{sec:choref}) that checks sufficient
conditions for well-formedness.
Secondly, we extend g-choreographies with our refinement mechanism
(\cref{sec:refinement}).
A key design choice of our framework is to ground refinements on the concept
\emph{refinable interactions}.
Inspired by the action refinement mechanism of process algebra, we consider
refinable g-choreographies those where refinable interactions may occur.
Refinable g-choreographies play the role of incomplete specifications
where, by repeated replacements of refinable interactions, one can
incrementally attain a fully specified global view.

One problem that may arise in this process is that refinements could
spoil well-formedness and hence compromise realisability of global
views.
To avoid this we extend the typing discipline for non-refinable
g-choreographies to refinable ones and show that the replacement of a
refinable interaction with a g-choreography typable with the same type
ensures realisability.

We discuss related work and draw some conclusions in \cref{sec:conc}.


\section{Background}\label{sec:bkg}

We recall basic notions of \emph{event structures} used in
\cref{sec:choref} to give semantics to refinable choreographies.
Event structures model concurrency in terms of partial orders of
labelled events.
We focus our attention on event structures of communication events.
Let $\ptpset$ be a set of \emph{participants} (ranged over by $\ptp$,
$\ptp[B]$, etc.) and $\msgset$ be a set of \emph{messages} (ranged
over by $\msg$, $\msg[x]$, etc.).
We take $\ptpset$ and $\msgset$ disjoint.
Let
\[
  \ChanSet = (\PartSet \times \PartSet) \setminus \Set{(\p,\p) \sst \p \in \PartSet}
  \qquad
  \lset^! = \ChanSet \times \Set{!} \times \msgSet
  \qquad
  \lset^? =
  \ChanSet \times \Set{?} \times \msgSet
\]
be the sets of \emph{channels}, output labels, and input labels
respectively.
We write $\aout \in \lset^!$ and $\ain \in \lset^?$ instead of
$((\p,\q), !, \msg) \in \lset^!$ and $((\p,\q), ?, \msg) \in \lset^?$.
The \emph{subject} of a label $\al$, written $\esbj[\al]$, is defined
as $\esbj[\aout] = \esbj[{\ain[B][A]} ]= \{\p\}$.
The elements of $\lset = \lset^! \cup \lset^?$ (ranged over by $\al$)
will be used to label the events of our event structures.
The \emph{co-action} of $\al \in \lset$ is defined as
$\co{\aout} = \ain$ and $\co{\ain} = \aout$ and extends element-wise
on sets of actions.

\begin{definition}[Event structures]
  An \emph{event structure} labelled over $\lset$ (shortly {\em event
	 structure}) is a tuple $\eset = (\aE, \leq, \conflict, \lambda)$
  where
  \begin{itemize}
  \item $\aE$ a set of \emph{events}
  \item $\leq \;\subseteq \; \aE \times \aE$ a partial order, the
	 \emph{causality} relation
  \item $\conflict \;\subseteq \;\aE \times \aE$ a symmetric and
	 irreflexive relation, the \emph{conflict} relation,
  \item $\lambda: \aE \to \lset$ a \emph{labelling} mapping.
  \end{itemize}
  are such that
  \begin{itemize}
  \item each event has only finitely many predecessors, namely
	 $\forall \ae \in \aE.\; \Set{\ae' \in \aE \mid \ae' \leq \ae}$ is
	 finite, and
  \item conflicts are hereditary, namely
	 $\forall \ae, \ae', \ae'' \in \aE.\; \ae \conflict \ae' \And \ae'
	 \leq \ae'' \Then \ae \conflict \ae''$
  \end{itemize}
  If $\eset = (\aE, \leq, \conflict, \lambda)$ is an event structure
  then $\min (\eset), \max (\eset) \subseteq \aE $ are the minimal and
  the maximal elements in the poset $(\aE, \leq)$.
  We define $\emptyev = (\emptyset, \emptyset, \emptyset, \emptyset)$
  as the empty event structure, where $\lambda_\emptyset = \emptyset$
  is the empty mapping.
\end{definition}
Notice that if $\eset \neq \emptyev$ then minimal elements do exist,
while this is not necessarily the case for maximal ones.
We depict event structures following the customary representation of
the literature~\cite{win81} as the diagram \eqref{eq:es} below;
instead of events though, we prefer to use their labels, for instance:
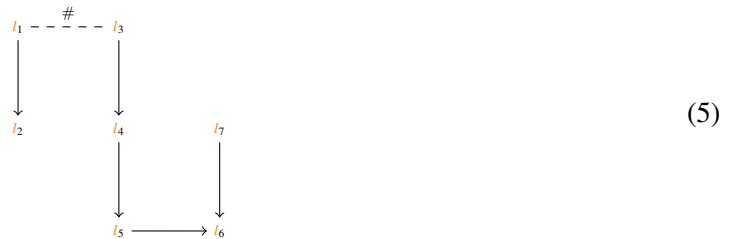
\begin{align}\label{eq:es}
  \evstr[1]{
	 \node (l) {$\al_1$};
	 \node[right = of l] (r) {$\al_3$};
	 \node[below = of l] (l?) {$\al_2$};
	 \node[below = of r] (r?) {$\al_4$};
	 \node[below = of r?] (s?) {$\al_5$};
	 \node[right = of r?] (s1) {$\al_7$};
	 \node[below = of s1] (s1?) {$\al_6$};
	 \evconflict l r;
	 \evleq l {l?};
	 \evleq r {r?};
	 \evleq {r?} {s?};
	 \evleq {s1} {s1?};
	 \evleq {s?} {s1?};
  }
\end{align}
represents an event structure with events $\ae_1, \ldots, \ae_7$ (not
represented in the diagram above) and labelled respectively by
$\al_1,\ldots,\al_7$ where
\begin{itemize}
\item the event $\ae_1$ precedes $\ae_2$ (i.e., $\ae_1 \leq \ae_2$ in
  the partial order of the event structure)
\item events $\ae_1$ and $\ae_3$ are in conflict; recall that the
  conflict relation is hereditary, hence $\ae_1$ and $\ae_2$ are in
  conflict with all other events but $\ae_7$
\item events $\ae_2$ and $\ae_6$ are maximal; the latter follows
  both $\ae_5$ and $\ae_7$
\item events $\ae_5$ and $\ae_7$ are independent of each other
(actually, $\ae_7$ is independent of any event but $\ae_6$).
\end{itemize}
In our diagrams we adopt the implicit assumption that each occurrence
of a label correspond to a different event; for instance, in the
diagram \eqref{eq:es}, even if two labels, say $\al_1$ and $\al_7$
were equal, the corresponding events would be distinct (i.e.,
$\ae_1 \neq \ae_7$).

An event structure induces a natural order and conflict relations
on the events performed by each participant.
More precisely, the \emph{projection} of an event structure
$\eset = (\aE, \leq, \conflict, \lambda)$ on a participant
$\ptp \in \PartSet$ is the structure
\[\Proj{\eset}{\ptp} = (\Proj{\aE}{\ptp}, \Proj{\leq}{\ptp}, \Proj{\conflict}{\ptp}, \Proj{\lambda}{\ptp})\]
where
\begin{description}
\item $\Proj{\aE}{\ptp} = \Set{\ae \in \aE \mid \esbj[\lambda(\ae)] = \ptp}$
\item $\Proj{\leq}{\ptp} = \leq \cap\, (\Proj{\aE}{\ptp})^2$ and $\Proj{\conflict}{\ptp} = \conflict \cap\, (\Proj{\aE}{\ptp})^2$
\item $\Proj{\lambda}{\ptp} = \lambda \big|_{\Proj{\aE}{\ptp}}$, namely the restriction of $\lambda$ to $\Proj{\aE}{\ptp}$
\end{description}
Trivially, the induced relations form an event structure.
\begin{lemma}\label{lem:projection}
  If $\eset$ is an event structure and $\ptp \in \PartSet$ then
  $\Proj{\eset}{\ptp}$ is an event structure.
\end{lemma}
\begin{proof}
  Immediate since $\Proj{\leq}{\ptp} \subseteq \; \leq$ and
  $\Proj{\conflict}{\ptp} \subseteq \conflict$.
\end{proof}

We now define a few operations instrumental to our technical
development.
Let $\eset_0 = (\aE_0, \leq_0, \conflict_0, \lambda_0)$ and
$\eset_1 = (\aE_1, \leq_1, \conflict_1, \lambda_1)$ be labelled event
structures.

The product operation $\_ \Tensor \_$ yields the disjoint union of
event structures preserving their orders, conflicts, and labellings;
it is define as
\[
\eset_0 \Tensor \eset_1 = (\aE_0 \uplus \aE_1, \leq, \conflict,
\lambda)
\] 
where writing $\inj_i: \aE_i \to \aE_0 \uplus \aE_1$ for the injections, we set
\[
  \inj_i \ae \leq \inj_j \ae' \iff i = j \And \ae \leq_i \ae'
  \qquad
  \inj_i \ae \conflict \inj_j \ae' \iff i = j \And \ae \conflict_i \ae'
  \qquad
  \lambda(\inj_i \ae) = \lambda_i(\ae)
\]
The sum $\sum_{i\in I} \eset_i$ yields the disjoint union of a family
$\Set{\eset_i}_{i\in I}$ of event structures
$\eset_i = (\aE_i, \leq_i, \conflict_i, \lambda_i)$ preserving their
orders and labellings while introducing conflicts among events of
different members of the family; it is defined as the event structure
$(\biguplus_{i\in I} \aE_i, \leq, \conflict, \lambda)$ where, writing
$\inj_i: \aE_i \to \biguplus_{i\in I} \aE_i$ for the injections, the
following hold:
\[
  \inj_i \ae \leq \inj_j \ae' \iff i = j \And \ae \leq_i \ae'
  \qquad
  \inj_i \ae \conflict \inj_j \ae' \iff i \neq j \Or (i = j \And \ae
  \conflict_i \ae')
  \qquad
  \lambda(\inj_i \ae) = \lambda_i(\ae)
\]
In particular we write 
\[ \eset_0 + \eset_1 = \sum_{i \in \Set{0,1}}\eset_i \qquad \text{and}
  \qquad \sum_{i\in I} \eset = \sum_{i\in I} \eset_i \quad \text{where }
  \eset_i = \eset \text{ for all } i\in I
\]

\begin{lemma}\label{lem:tensorAndSum}
If $\eset_0, \eset_1$ are event structures and $\Set{\eset_i}_{i\in I}$ is a family of event structures then 
\[
  \eset_0 \Tensor \eset_1 \qqand \sum_{i\in I} \eset_i \qquad \text{ are event structures.}
\]
\end{lemma}

\begin{definition}[Configuration domain]
  If $\eset = (\aE, \leq, \conflict, \lambda)$ is an event structure
  a set of events $x \subseteq \aE$ is a \emph{configuration} if
  \begin{enumerate}
  \item $\ae \leq \ae' \And e' \in x \Then \ae \in x$ ($x$ is downward
	 closed)
  \item $ \forall \ae, \ae' \in x.\; \neg(\ae \conflict \ae')$ ($x$ is
	 consistent)
  \end{enumerate}
  Let $C = \Set{x \subseteq \aE \sst x \text{ a configuration}}$; the
  \emph{domain of configurations} of $\eset$ is the poset
  $\Domain{\eset} = (C, \subseteq)$.
  We say that $x \in C$ is \emph{maximal} if it is such in
  $\Domain{\eset}$: $\MaxC{\eset}$ is the set of maximal
  configurations.
\end{definition}

Being conflict-free and maximal, configurations in $\MaxC{\eset}$ correspond to branches of events of $\eset$.


\section{Well-formedness by Typing}\label{sec:choref}
We formalise global views of choreographies as
g-choreographies~\cite{gt16,gt18}.
Although we maintain the original syntax, we provide a new semantics
of g-choreographies based on event structures.
This is instrumental to identify a simple notion of well-formedness
that can be statically checked.
\subsection{Global Choreographies}

\cref{{def:refgg}} introduces \emph{ global choreographies}.
The syntax of a g-choreography is given by the grammar below
that we borrow from~\cite{gt18}.
\begin{definition}[Global Choreographies]\label{def:refgg}
  The set $\gset$ of \emph{global choreographies
	 (g-choreographies for short)} consists of the terms $\aG$ derived
  by the grammar
  \begin{align}
    \aG \bnfdef & \gempty & \text{empty} \label{eq:chorempty}
    \\
    \bnfmid & \gint & \text{interaction} \label{eq:chorint}
    \\
    \bnfmid & \gseq[] & \text{sequential} \label{eq:chorseq}
    \\
    \bnfmid & \gpar & \text{parallel} \label{eq:chorpar}
    \\
    \bnfmid & \gcho & \text{choice} \label{eq:chorcho}
  \end{align}
  such that $\ptp \neq \ptp[B]$ in interactions \eqref{eq:chorint}
  We let $\PartSet(\aG)$ be the set of participants occurring in $\aG$.
\end{definition}
Besides the empty choreography $\gempty$, the syntax of
\cref{def:refgg} allows us to specify choreographies whose basic
elements are interactions $\gint[]$
which represent that participant \p\ sends message $\msg$
to participant \q, which in turn should receive it.
Finally, g-choreographies can be composed sequentially, in parallel,
and non-deterministically.
The syntax in~\cite{gt18} encompasses iterative g-choreographies
which we drop for simplicity.
Adding iteration can be done following standard techniques at the cost
of a substantial increase of the technical complexity.

\begin{example}\rm
  The term in \eqref{eq:err} in \cref{sec:intro} is a g-choreography.
  \finex
\end{example}
We now give the semantics of g-choreographies in terms of event
structures.
To this purpose, note that not every $\aG$ is \quo{meaningful} because
$\aG$ can specify protocols where the behaviour of some participants, say
\q, depends on choices made by others that are not properly propagated
to \q.
The following example illustrates this.
\begin{example}\label{ex:notwb}\rm
  The g-choreography
  $\aG =
  \gcho[][{\gseq[][{\gint[][@][m][c]}][{\gint[][b][m][c]}]}][{\gseq[][{\gint[][@][n][c]}][{\gint[][b][n][c]}]}]$
  specifies a protocol where \p\ decides whether to send
  $\msg[m]$ or $\msg[n]$ to \p[c].
  In either case \q\ should mimic \p\ and send the same message to
  \p[c].
  However, in a distributed implementation of this protocol \q\ is
  oblivious of the decision of \p; hence, e.g., \q\ could send message
  $\msg[m]$ while \p\ decided to send message $\msg[n]$.
  \finex
\end{example}

To mitigate the problem above, we give \emph{well-formedness}
conditions that rule out meaningless g-choreographies.
We start with \emph{well-branchedness}.
\begin{definition}[Well-branchedness]\label{def:wb}
  Event structures $\eset_0$ and $\eset_1$ are
  \emph{well-branched} (in symbols $\wb[\eset_0][\eset_1]$) if, for
  $\eset = \eset_0 + \eset_1 = (\aE, \leq, \conflict, \lambda)$, the
  following two conditions hold:
  \begin{align*}
    \textit{determined choice:\quad }
	 &
		\forall \q \in \ptpset \qst
		(\Proj{\eset_0}{\q} = \emptyev \iff \Proj{\eset_1}{\q} = \emptyev) \And
		\\ & \qquad
		\forall \ae, \ae' \in \min (\Proj \aE \q) \qst \ae \conflict \ae' \Then \lambda(\ae) \neq \lambda(\ae')
    \\
    \textit{unique selector:\quad } & \exists \p \in \ptpset \qst \emptyset \neq \lambda(\min (\Proj \aE \p)) \subseteq \{\al \in \lset^! \sst \esbj[\al] = \p\} \And
    \\ & \qquad \forall \q \neq \p \in \ptpset \qst \lambda(\min (\Proj \aE \q)) \subseteq \{\al \in \lset^? \sst \esbj[\al] = \q\}
  \end{align*}
  We dub \emph{active} the unique participant $\p$ satisfying
  the second condition and \emph{passive} the others.
\end{definition}

Well-branchedness is akin to the conditions on behavioural types that
enforce choice determinacy.
Namely, each choice is determined by a unique participant, dubbed
selector, which starts to send messages to the others and that any
non-selector participant becomes aware of the choice taken by the
selector just because of the messages received on a branch.

We re-cast the notion of \emph{well-forkedness} in~\cite{gt18} in
terms of event structures.
\begin{definition}[Well-forkedness]\label{def:wf}
  Two event structures $\eset = (\aE, \leq, \conflict, \lambda)$ and
  $\eset' = (\aE', \leq, \conflict', \lambda')$ are \emph{well-forked}
  (in symbols $\wf[\eset][\eset']$) if
  $\lambda(\aE) \cap \lambda'(\aE') = \emptyset$.
\end{definition}
As well-branchedness, the parallel composition of g-choreographies is
subject to some conditions.
As observed in~\cite{gt18}, \quo{confusion} may arise when different
threads of participants exchange the same message: the message meant to
be received by a thread is received by the other.
If this happens there is a violation of the causal order of the
events.
The next example illustrates the problem.
\begin{example}\rm
  The g-choreography
  $\gpar[][{\gseq[][{\gint[]}][{\gint[][b][m][c]}]}][{\gseq[][{\gint[]}][{\gint[][b][n][d]}]}]$
  is not well-forked because the interaction between \q\ and \p[c]
  should start \emph{after} the \quo{left thread} of \p\ had sent message
  $\msg$.
  However, it could happen that the left thread of \q\ receives the
  message sent by the right thread of \p\ so violating the
  specification.
  And likewise for the \quo{right threads}.
  \finex
\end{example}

\begin{definition}[Sequential composition]\label{def:seq}
Let $\eset$, $\eset'$ be event structures and
\[(\aE'', \leq'', \conflict'', \lambda'') = \eset \Tensor \sum_{x \in
  \MaxC{\eset}} \eset_x'\] 
where the structures
$\eset_x' = (\aE_x, \leq_x, \conflict_x, \lambda_x)$ are disjoint
copies of $\eset'$, then
\[
  \rseq[\eset][\eset'] = (\aE'', \leq'' \; \cup \; \bigcup_{x \in \MaxC{\eset}}
  \Set{(\ae,\ae') \in x \times \aE_x \sst
	 \esbj[\lambda'' (\ae)] = \esbj[\lambda'' (\ae')]}, \conflict'', \lambda'').
\]
\end{definition}
The intuition of the definition of $\rseq[\eset][\eset']$ is that any
branch $x \in \MaxC{\eset}$ of $\eset$ is concatenated to a (pairwise
incompatible) copy of $\eset'_x$, where events in $\eset$ cause those
of $\eset'_x$ with labels having the same subject. Admittedly, in the
context of Definition \ref{def:interpretation} this is unnecessarily
abstract, since an event structure $\eset$ interpreting a
g-choreography is finite, and hence $\MaxC{\eset}$ and any of its
elements are such: hence any $x \in \MaxC{\eset}$ includes a finite
subset of maximals with respect to $\leq_{\eset}$.
However the definition applies to infinite structures as well.

\begin{lemma}\label{lem:seq}
If $\eset, \eset'$ are event structures, then $\rseq[\eset][\eset']$ is an event structure.
\end{lemma}

We can now give a denotational semantics of g-choreographies.
We require our semantics to be defined only on g-choreographies
amenable of being realised by distributed components satisfying
the following requirements:
\begin{itemize}
\item no extra components: each component uniquely corresponds to a
  participant of the g-choreography
\item no extra communications: each communication among the components
  uniquely corresponds to some communication events of the (semantics
  of the) g-choreography
\end{itemize}
These requirements impose that the communication behaviour of a
realisation of a g-choreography faithfully reflects the communication
events of the g-choreography.

\begin{definition}[Semantics]\label{def:interpretation}
  Let $\aG$ be a g-choreography.
  The \emph{semantics} $\Sem{\aG}$ of $\aG$ is the partial mapping assigning
  an event structure to $\aG$ according to the following inductive clauses:
  \[\begin{array}{rcl}
		\Sem{ \gempty} & = & \emptyev
		\\ [2mm]
		\Sem{\gint} & = &  (\Set{\ae_1, \ae_2}, \Set{\ae_1 < \ae_2},
								\emptyset, \Set{\ae_1 \mapsto \aout, \; \ae_2
								\mapsto \ain[B][A]})
		\\ [2mm]
		\Sem{\gseq[]} & = & \rseq [\Sem{\aG}] [\Sem{\aG'}]
		\\ [2mm]
		\Sem{\gpar} & = & \left\{ \begin{array}{ll}
											 \Sem{\aG} \Tensor \Sem{\aG'} & \text{if } \wf[\Sem{\aG}][\Sem{\aG'}]
											 \\ [1mm] \bot & \text{otherwise}
										  \end{array} \right.
		\\ [6mm]
		\Sem{\gcho} & = & \left\{ \begin{array}{ll}
											 \Sem{\aG} + \Sem{\aG'} & \text{if } \wb[\Sem{\aG}][\Sem{\aG'}]
											 \\ [1mm]
											 \bot & \text{otherwise}
										  \end{array} \right.
	 \end{array}
  \]
where if either $\Sem{\aG}$ or $\Sem{\aG'}$ is $\bot$, then $\rseq [\Sem{\aG}] [\Sem{\aG'}]$, $\Sem{\aG} \Tensor \Sem{\aG'}$ and $\Sem{\aG} + \Sem{\aG'}$ are all equal to $\bot$.
Finally we say that $\aG$ is {\em well-formed} if $\Sem{\aG} \neq \bot$.
\end{definition}

We say that a g-choreography $\aG$ is \emph{well-formed} when each
choice subterm of $\aG$ is well-branched and each parallel subterm of
$\aG$ is well-forked.

\begin{example}\rm
  Let us spell out the semantics of the g-choreography
  $\aG =
  \gcho[][{\gint[][c][md][s]}][{\gseq[][{\gint[][c][req][s]}][{\gseq[][{\gint[][s][stats][c]}][{\gint[][s][done][c]}]}]}]$
  obtained by the refinement \eqref{ex:tieref} (cf. \cref{sec:intro})
  with the further refinement of $\refgint[c][md][s]$ with its ground
  counterpart $\gint[][c][md][s]$.
  By \cref{def:interpretation}, $\Sem \aG$ is defined if
  $\wb[\Sem{\gint[][c][md][s]}][\Sem{\gseq[][{\gint[][c][req][s]}][{\gseq[][{\gint[][s][stats][c]}][{\gint[][s][done][c]}]}]}]$
  holds.
  We now verify that this is the case.
  By definition, we have:
  \[
	 \Sem{\gint[][c][md][s]} =
	 \evstr{
		\node (l) {$\aout[c][s][][md]$};
		\node[below = of l] (r) {$\ain[c][s][][md]$};
		\evleq l r;
	 }
  \qqand
  \Sem{\gseq[][{\gint[][c][req][s]}][{\gseq[][{\gint[][s][stats][c]}][{\gint[][s][done][c]}]}]} =
	 \evstr{
		\node (l) {$\aout[c][s][][req]$};
		\node[below = of l] (r?) {$\ain[c][s][][req]$};
		\node[below = of r?] (s) {$\aout[s][c][][stats]$};
		\node[below = of s] (s?) {$\ain[s][c][][stats]$};
		\node[right = of s] (s1) {$\aout[s][c][][done]$};
		\node[below = of s1] (s1?) {$\ain[s][c][][done]$};
		\evleq l {r?};
		\evleq {r?} s;
		\evleq s {s?};
		\evleq s {s1};
		\evleq {s1} {s1?};
		\evleq {s?} {s1?};
	 }
  \]
  The sum operation on event structures introduces conflicts between the events in
  $\Sem{\gint[][c][md][s]}$ and those in
  $\Sem{\gseq[][{\gint[][c][req][s]}][{\gseq[][{\gint[][s][stats][c]}][{\gint[][s][done][c]}]}]}$,
  hence:
  \[
	 \eset = \Sem{\gint[][c][md][s]} +   \Sem{\gseq[][{\gint[][c][req][s]}][{\gseq[][{\gint[][s][stats][c]}][{\gint[][s][done][c]}]}]} =
	 \evstr{
		\node (l) {$\aout[c][s][][md]$};
		\node[right = of l] (r) {$\aout[c][s][][req]$};
		\node[below = of l] (l?) {$\ain[c][s][][md]$};
		\node[below = of r] (r?) {$\ain[c][s][][req]$};
		\node[below = of r?] (s) {$\aout[s][c][][stats]$};
		\node[below = of s] (s?) {$\ain[s][c][][stats]$};
		\node[right = of s] (s1) {$\aout[s][c][][done]$};
		\node[below = of s1] (s1?) {$\ain[s][c][][done]$};
		\evconflict l r;
		\evleq l {l?};
		\evleq r {r?};
		\evleq {r?} s;
		\evleq s {s?};
		\evleq s {s1};
		\evleq {s1} {s1?};
		\evleq {s?} {s1?};
	 }
  \]
  (recall that conflicts are hereditary, hence it is enough to put
  only minimal events in conflict).
  Now we look at the projections on
  \p[c] and on \p[s] of $\eset$:
  \[
  	 \Proj{\Sem \aG}{\p[c]} =
	 \evstr{
		\node (l) {$\aout[c][s][][md]$};
		\node[right = of l] (r) {$\aout[c][s][][req]$};
		\node[below = of r] (s?) {$\ain[s][c][][stats]$};
		\node[below = of s?] (s1?) {$\ain[s][c][][done]$};
		\evconflict l r;
		\evleq r {s?};
		\evleq {s?} {s1?};
	 }
    \qqand
	 \Proj{\Sem \aG}{\p[s]} =
	 \evstr{
		\node (l) {$\ain[c][s][][md]$};
		\node[right = of l] (r) {$\ain[c][s][][req]$};
		\node[below = of r] (s?) {$\aout[s][c][][stats]$};
		\node[below = of s?] (s1?) {$\aout[s][c][][done]$};
		\evconflict l r;
		\evleq r {s?};
		\evleq {s?} {s1?};
	 }
  \]
  where in $\Proj{\eset}{\p[c]}$ the minimal events are in conflict
  because they are in conflict in $\eset$ by construction and in
  $\Proj{\Sem \aG}{\p[s]}$ they are in conflict because in $\eset$ the
  events inherit previous conflict.
  It is easy to verify that the conditions of \cref{def:wb} are
  therefore verified.
  \finex
\end{example}

\begin{example}\label{ex:undef}\rm
  Recall the g-choreography \eqref{eq:err} in \cref{sec:intro} which
  we rewrite as $\aG_\mathit{err} = \gcho[]$ where
  \[
        \aG =
        \gseq[][{\gint[][c][md][b]}][{\gint[][b][md][s]}]
        \qqand
        \aG' = \gseq[][{\gint[][c][req][s]}][{\gint[][s][done][c]}]
  \]
  Let us show that $\Sem{\aG_\mathit{err}} = \bot$.
  In fact,
  \[
  \Sem{\aG} =
  \evstr{
    \node (l) {$\aout[c][b][][md]$};
    \node[below = of l] (r) {$\ain[c][b][][md]$};
    \node[below = of r] (s?) {$\aout[b][s][][md]$};
    \node[below = of s?] (s1?) {$\ain[b][s][][md]$};
    \evleq l r;
    \evleq r {s?};
    \evleq {s?} {s1?};
  }
  \qqand
  \Sem{\aG'} =
  \evstr{
    \node (l) {$\aout[c][s][][req]$};
    \node[below = of l] (r) {$\ain[c][s][][req]$};
    \node[below = of r] (s?) {$\aout[s][c][][done]$};
    \node[below = of s?] (s1?) {$\ain[s][c][][done]$};
    \evleq l r;
    \evleq r {s?};
    \evleq {s?} {s1?};
  }
  \qqand[hence]
  \Sem{\aG} +  \Sem{\aG'} =
  \evstr{
    \node (l) {$\aout[c][b][][md]$};
    \node[below = of l] (l1) {$\ain[c][b][][md]$};
    \node[below = of l1] (l2) {$\aout[b][s][][md]$};
    \node[below = of l2] (l3) {$\ain[b][s][][md]$};
    \node[right = of l] (r) {$\aout[c][s][][req]$};
    \node[below = of r] (r1) {$\ain[c][s][][req]$};
    \node[below = of r1] (r2) {$\aout[s][c][][done]$};
    \node[below = of r2] (r3) {$\ain[s][c][][done]$};
    \evleq{l} {l1};
    \evleq{l1}{l2};
    \evleq{l2}{l3};
    \evleq{r}{r1};
    \evleq{r1}{r2};
    \evleq{r2}{r3};
    \evconflict l r;
  }
  \]
  It is easy to check that the determined choice condition of
  \cref{def:wb} does not hold for \q.
  \finex
  \end{example}


\subsection{Typing well-formedness}\label{sec:typing}

\newcommand{\mdmsg}{\mathsf{\colorMsg{md}}}
\newcommand{\reqmsg}{\mathsf{\colorMsg{req}}}
\newcommand{\donemsg}{\mathsf{\colorMsg{done}}}

\newcommand{\der}{\vdash}
\newcommand{\typetp}[2]{\langle #1, #2 \rangle}

By \cref{def:interpretation}, non-wellformed g-choreographies $\aG$
are meaningless, i.e. $\Sem{\aG} = \bot$.
However, it is too expensive to check that $\Sem{\aG} \neq \bot$ via a
direct inspection of the event structure $\Sem{\aG}$.
Also, it is hard to see how to extend the semantics to refinable
communications, as their intended meaning is an infinite set of
possible realisations by concrete g-choreographies.
To circumvent this difficulty we formalise sufficient conditions for
well-formedness via a typing system.

Our typing discipline assigns to a g-choreography type
$\conf{\minP, \maxP}$ where $\minP \subseteq \lset$ and
$\maxP \subseteq \lset$.
Intuitively, $\minP$ and $\maxP$ are respectively the labels of the
first and last events in the g-choreography.
Our judgements have the form
\[
  \tj  \CxtPi \aG \minP \maxP
\]
and their intended meaning is: the g-choreography $\aG$ has a defined
semantics and it has type $\conf{\minP, \maxP}$ under the assumption
that its participants are those in $\CxtPi \subseteq \PartSet$.
\begin{remark}
  We could avoid the use of the context $\CxtPi$ in judgements; we however
  prefer to explicitly list relevant participants for clarity.  
\end{remark}
In the following we illustrate the typing rules, by defining side conditions,
explaining the notation, and relating the rules to the semantics of
choreographies in \cref{sec:choref}.
\cref{fig:typing} collects all the rules for convenience; in
commenting the rules, we motivate their 
soundness (cf. Theorem~\ref{thr:soundness}).



\begin{figure}[ht]
  \[\begin{array}{c@{\hspace{10mm}}c}
		\arule{}{\tj \emptyset \gempty \emptyset \emptyset}{t-emp}
		&
		\arule{\minP = \maxP = \{\aout, \ain\}}{\tj  {\{\p,\q\} } \gint \minP \maxP}{t-int}
		\\[30pt]
		\multicolumn 2 c {
		\arule{
		\tj  {\CxtPi_1} {\aG_1} {\minP_1} {\maxP_1}
		\qquad
		\tj  {\CxtPi_2} {\aG_2} {\minP_2} {\maxP_2}
		}{
		\tj  {\CxtPi_1\cup \, \CxtPi_2} {\gseq[][\aG_1][\aG_2]}
		{ \minP_1\cup (\minP_2 - \Pi_1)} {  \maxP_2 \cup (\maxP_1 -
		\Pi_2)}
		}{
		t-seq
		}
		}
		\\[30pt]
		\multicolumn 2 c {
		\arule{
		\tj  {\CxtPi_1} {\aG_1} {\minP_1} {\maxP_1}
		\qquad
		\tj  {\CxtPi_2} {\aG_2} {\minP_2} {\maxP_2}
		\qquad \CxtPi_1 \cap \CxtPi_2 = \emptyset
		}{
		\tj  {\CxtPi_1\cup\CxtPi_2} {\gpar[][\aG_1][\aG_2]}
		{\minP_1\cup\minP_2} {\maxP_1\cup\maxP_2}
		}{
		t-par
		}
		}
		\\[30pt]
		\multicolumn 2 c {\arule{
		\tj  \CxtPi {\aG_1} {\minP_1} {\maxP_1}
		\qquad
		\tj  \CxtPi {\aG_2} {\minP_2} {\maxP_2}
		\qquad
		\compch {\minP_1} {\minP_2}\CxtPi
		}{
		\tj  \CxtPi {\gcho[][\aG_1][\aG_2]} {\minP_1\cup\minP_2}
		{\maxP_1\cup\maxP_2}
		}{
		t-ch
		}
		}
	 \end{array}
\]
\caption{Typing rules for g-choreographies.}
\label{fig:typing}
\end{figure}



\paragraph*{Interaction}
Define the mapping $\widehat{\_}: 2^\lset \to \PartSet \to 2^{\lset}$
by $\widehat{L}(\ptp) = \Set{\al \in L \mid \esbj[\al] = \ptp}$.
Then
\[
  L = \bigcup_{\ptp \in \PartSet} \widehat{L}(\ptp)
\]
so that we can see any $L \subseteq \lset$ as a family of sets of
(labels of) actions indexed over $\PartSet$.
Note that if $L$ is finite then $\widehat{L}(\ptp) \neq \emptyset$ for
finitely many $\ptp$.
Now, inspecting the rule for interaction:
\[\arule{
	 \minP = \maxP = \{\aout, \ain\}
  }{
	 \tj  {\{\p,\q\} } \gint \minP \maxP
  }{
	 t-int
  }
\]
we see that $\Set{\p,\q} = \PartSet(\gint)$; also we know that $\gint$
has a defined semantics (recall that $\PartSet(\aG)$ are the
participants occurring in $\aG$):
\[\begin{array}{c}
  \Sem{\gint} =  (\Set{\ae_1, \ae_2}, \Set{\ae_1 < \ae_2}, \emptyset,
  \Set{\ae_1 \mapsto \aout, \; \ae_2 \mapsto \ain[B][A]})
  \qquad\text{hence}
  \\[1em]
  \widehat{\minP}(\p) = \widehat{\maxP}(\p) = \Set{\aout} = \min\, (\Proj{\Sem{\gint}}{\p})
  \qand
  \widehat{\minP}(\p) = \widehat{\maxP}(\q) = \Set{\ain[B][A]} = \min\, (\Proj{\Sem{\gint}}{\q})
\end{array}\]
The distinction among minimal and maximal elements in a singleton
poset is clearly immaterial; it becomes sensible in case of the
subsequent rules.
\paragraph*{Sequential composition}
If $L \subseteq \lset$ and $\Pi \subseteq \ptpset$ then set
$L - \Pi = \Set{\al \in L \mid \esbj[\al] \not \in \Pi}$. Then the
rule is:
\[
\arule{
\tj  {\CxtPi_1} {\aG_1} {\minP_1} {\maxP_1}
\qquad
\tj  {\CxtPi_2} {\aG_2} {\minP_2} {\maxP_2}
}{
\tj  {\CxtPi_1\cup \, \CxtPi_2} {\gseq[][\aG_1][\aG_2]} 
	{ \minP_1\cup (\minP_2 - \Pi_1)} {  \maxP_2 \cup (\maxP_1 - \Pi_2)}
}{
t-seq
}
\]
By induction $\Sem{\aG_i} = \eset_i \neq \bot$ for $i = 1, 2$, hence
$\Sem{\gseq[][\aG_1][\aG_2]} = \rseq [\Sem{\aG_1}] [\Sem{\aG_2}]$ is
defined.
Let $\eset_i = \Sem{\aG_i}$ and, for each indexed over $x$ in the set
$\MaxC{\eset_1}$ of maximal configurations of $\eset_1$ , $\eset_x$ be
disjoint event structures isomorphic to $\eset_2$.
Then by definition, the order relation $\leq_{\textsf{seq}}$ of
$\rseq [\Sem{\aG_1}] [\Sem{\aG_2}]$ is the relation;
\[
  \leq_1 ~ \cup \bigcup_{x \in \MaxC{\eset_1}} \left( \leq_x
	 \; \cup ~\Set{ (\ae_1, \ae_2) \in x \times \aE_x \sst 
		\esbj[\lambda_1(\ae_1)] = \esbj[\lambda_x(\ae_2)]}
  \right)
\]
where $\leq_x$ is the order relation of $\eset_x$, and
$\lambda_x(\ae_2) = \lambda_2(\ae_2)$ by construction.  If $\ae$ is
minimal with respect to $\leq_{\textsf{seq}}$ then either it is such
with respect to $\leq_1$, or $\ae \in \aE_x$ is minimal with respect to $\leq_x$ and
$\esbj[\lambda_1(\ae')] \neq \esbj[\lambda_x(\ae)]$ for all maximal
$\ae' \in \aE_1$. If we consider the projection
$\Proj{\Sem{\gseq[][\aG_1][\aG_2]}}{\p}$ for any participant $\p$ of
$\gseq[][\aG_1][\aG_2]$, then by definition all event labels have the
same subject $\p$.
So if $\p \in \PartSet(\aG_1)$ then each maximal configuration of
$\Sem{\aG_1}$ has an event whose label has subject \p, hence
$\min (\Proj{\Sem{\gseq[][\aG_1][\aG_2]}}{\p})$ are exactly the
minimal events in $\min (\Proj{\eset_1} {\p})$.
Otherwise,
$\min (\Proj{\Sem{\gseq[][\aG_1][\aG_2]}}{\p}) = \min (\Proj{\eset_2}
{\p})$.
By induction we know that $\Pi_i = \PartSet(\aG_i)$ and
$\widehat{\minP_i}(\p) = \min (\Proj{\eset_i} {\p})$.
Hence
\[\begin{array}{lll}
	 \min (\Proj{\Sem{\gseq[][\aG_1][\aG_2]}}{\p}) & =
	 & \min (\Proj{\eset_1} {\p}) \cup \Set{\ae \in \min (\Proj{\eset_2} {\p}) \mid \p \not \in  \PartSet(\aG_1)}
	 \\ & =
	 & \widehat{\minP_1}(\p) \cup (\widehat{\minP_2 - \Pi_1} (\p))
	 \\ & =
	 & (\widehat{\minP_1 \cup (\minP_2 - \Pi_1)})(\p)
  \end{array}
\]
where of course either $\widehat{\minP_1}(\p)$ or
$\widehat{\minP_2 - \Pi_1}(\p)$ must be empty.
Similarly we get
$(\widehat{\maxP_2 \cup ( \maxP_1 - \Pi_2 )}) (\p) = \max
(\Proj{\Sem{\gseq[][\aG_1][\aG_2]}}{\p})$.
\begin{example}\label{exm:typingSeq}\rm
  Consider typing
  ${\gseq[][{\gint[][c][req][s]}][{\gint[][s][done][c]}]}$; then we
  have:
  \[
	 \prooftree
	 \prooftree
		\minP_1  = \maxP_1 = \{\aout[c][s][][\reqmsg], \ain[c][s][][\reqmsg]\}
	\justifies
		\Set{\ptp[C],\ptp[S]} \der \gint[][c][req][s] : \typetp{\minP_1}{ \maxP_1}
	\endprooftree
	\quad
	\prooftree
		\minP_2  = \maxP_2 = \{\aout[s][c][][\donemsg], \ain[s][c][][\donemsg]\}
	\justifies
		\Set{\ptp[C],\ptp[S]} \der \gint[][s][done][c]  : \typetp{\minP_2}{ \maxP_2}
		\endprooftree
		\justifies
		\Set{\ptp[C],\ptp[S]} \der \gseq[][{\gint[][c][req][s]}][{\gint[][s][done][c]}]  : \typetp{\minP_1}{ \maxP_2}
		\endprooftree
  \]
  because
  $\Set{\ptp[C],\ptp[S]} \cup \Set{\ptp[C],\ptp[S]} =
  \Set{\ptp[C],\ptp[S]}$,
  $\minP_1 \cup (\minP_2 - \Set{\ptp[C],\ptp[S]}) = \minP_1$ since
  $\minP_2 - \Set{\ptp[C],\ptp[S]} = \emptyset$, and
  $\maxP_2 \cup (\maxP_1 - \Set{\ptp[C],\ptp[S]}) = \maxP_2$ as
  $\maxP_1 - \Set{\ptp[C],\ptp[S]} = \emptyset$.  Similarly, typing
  $\gseq[][{\gint[][c][md][b]}][{\gint[][b][md][s]}]$ we obtain:
  \[
	 \prooftree
	 \prooftree
		\minP_3  = \maxP_3 = \{\aout[c][b][][\mdmsg], \ain[c][b][][\mdmsg]\}
	\justifies
		\Set{\ptp[C],\ptp[B]} \der \gint[][c][md][b] : \typetp{\minP_3}{ \maxP_3}
	\endprooftree
	\quad
	\prooftree
		\minP_4  = \maxP_4 = \{\aout[b][s][][\mdmsg], \ain[b][s][][\mdmsg]\}
	\justifies
		\Set{\ptp[B],\ptp[S]} \der \gint[][s][done][c]  : \typetp{\minP_4}{ \maxP_4}
		\endprooftree
		\justifies
		\Set{\ptp[B],\ptp[C],\ptp[S]} \der\gseq[][{\gint[][c][md][b]}][{\gint[][b][md][s]}] : \typetp{\minP_5}{ \maxP_5}
		\endprooftree
  \]
  where
  \[
	 \begin{array}{rccl}
		\minP_5 = \minP_3 \cup (\minP_4 - \Set{\ptp[C],\ptp[B]}) =  \minP_3 \cup \Set{\ain[b][s][][\mdmsg]}
		& = & \Set{ \aout[c][b][][\mdmsg], \ain[c][b][][\mdmsg] ,\ain[b][s][][\mdmsg]} & \mbox{and} \\ [1mm]
		\maxP_5 = \maxP_4 \cup (\maxP_3 - \Set{\ptp[B],\ptp[S]}) = \maxP_4 \cup \Set{\aout[c][b][][\mdmsg]} 
		& = &  \Set{ \aout[c][b][][\mdmsg], \aout[b][s][][\mdmsg], \ain[b][s][][\mdmsg] }
	 \end{array}
  \]
\end{example}

\paragraph*{Parallel composition}
The rule is:
\[
  \arule{
	 \tj  {\CxtPi_1} {\aG_1} {\minP_1} {\maxP_1}
	 \qquad
	 \tj  {\CxtPi_2} {\aG_2} {\minP_2} {\maxP_2}
	 \qquad \CxtPi_1 \cap \CxtPi_2 = \emptyset
  }{
	 \tj  {\CxtPi_1\cup\CxtPi_2} {\gpar[][\aG_1][\aG_2]}
	 {\minP_1\cup\minP_2} {\maxP_1\cup\maxP_2}
  }{
	 t-par
  }
\]
By induction we may suppose that, for $i = 1, 2$, $\CxtPi_i$ equals
$\PartSet(\aG_i)$ and $\Sem{\aG_i} = \eset_i \neq \bot$.
Hence, the condition $ \CxtPi_1 \cap \CxtPi_2 = \emptyset$ implies
that $\lambda_1(\aE_1) \cap \lambda_2(\aE_2) = \emptyset$, where
$\aE_i$ and $\lambda_i$ are the carrier and the labeling mapping of
$\eset_i$ respectively, so that
$\Sem{\gpar[][\aG_1][\aG_2]} = \Sem{\aG_1} \Tensor \Sem{\aG_2}$ is
defined.
Again by induction, for all participants $\p$ of $\aG_i$ we have that:
\[
  \widehat{\minP_i}(\p) = \min\, (\Proj{\Sem{\aG_i}}{\p})
  \qqand
  \widehat{\maxP_i}(\p) = \max\, (\Proj{\Sem{\aG_i}}{\p})
\] 
By definition of the tensor product, we know that
$\leq_{\eset_1\Tensor \eset_2}$ is just
$\leq_{\eset_1} \cup \leq_{\eset_2}$, which is a disjoint union (and
the same holds of the $\conflict_{\eset_1\Tensor \eset_2}$
relation).
Observing that
\[
  \widehat{\minP_1 \cup \minP_2}(\p) = \Set{\al \in \minP_1 \cup \minP_2 \mid \esbj[\al] = \p} =
  \Set{\al \in \minP_1 \mid \esbj[\al] = \p}  \cup \Set{\al \in \minP_2 \mid \esbj[\al] = \p} = \widehat{\minP_1}(\p) \cup \widehat{\minP_2}(\p) 
\] 
and similarly that
$\widehat{\maxP_1 \cup \maxP_2}(\p) = \widehat{\maxP_1}(\p) \cup
\widehat{\maxP_2}(\p)$, we conclude that, for all participants $\p$ of
$\gpar[][\aG_1][\aG_2]$:
\[
  \widehat{\minP_1 \cup \minP_2}(\p) = \min\, (\Proj{\Sem{\gpar[][\aG_1][\aG_2]}}{\p})
  \qqand
  \widehat{\maxP_1 \cup \maxP_2}(\p) = \max\, (\Proj{\Sem{\gpar[][\aG_1][\aG_2]}}{\p}).
\]

\paragraph*{Choice}

Two sets of labels $U, V \subseteq \lset$ are \emph{output uniform} if
$U \cap V = \emptyset$ and $U \cup V \subseteq \lset^!$; likewise, $U$
and $V$ are \emph{input uniform} if $U \cap V = \emptyset$ and
$U \cup V \subseteq \lset^?$.
Then the rule for typing choice is:
\[
  \arule{\tj  \CxtPi {\aG_1} {\minP_1} {\maxP_1}
	 \qquad
	 \tj  \CxtPi {\aG_2} {\minP_2} {\maxP_2}
	 \qquad
	 \compch {\minP_1} {\minP_2}\CxtPi
  }{
	 \tj  \CxtPi {\gcho[][\aG_1][\aG_2]} {\minP_1\cup\minP_2} {\maxP_1\cup\maxP_2}
  }{
	 t-ch
  }
\]
where the condition $\compch {\minP_1} {\minP_2}\CxtPi$ is defined by
the clauses:

\begin{enumerate}
\item \label{active} there is a unique $\p \in \Pi$ such that
  $\widehat{\minP_1} (\p)$ and $\widehat{\minP_2} (\p)$ are output
  uniform and both non-empty;
\item \label{passive} for all $\q \neq \p \in \Pi$,
  $\widehat{\minP_1} (\q)$ and $\widehat{\minP_2} (\q)$ are input
  uniform and $\widehat{\minP_1} (\q) = \emptyset$ if and only if
  $\widehat{\minP_2} (\q) = \emptyset$.
\end{enumerate}
By induction, for $i = 1,2$ $\Sem{\aG_i} \neq \bot$, the participants
of $\aG_i$ are $\Pi$ and
$\widehat{\minP_i}(\rr) = \min (\Proj{\eset_i} {\rr})$ and
$\widehat{\maxP_i}(\rr) = \max (\Proj{\eset_i} {\rr})$ for all
$\rr \in \Pi$.
Let $\Sem{\aG_i} = \eset_i =(\aE_i, \leq_i, \conflict_i, \lambda_i)$
and $\eset = \eset_0 + \eset_1$. By condition \ref{active} above, and
remembering the identification of $\ae$ with $\lambda_i(\ae)$, we have
that
$\widehat{\minP_i} (\p) = \min (\Proj{\eset_i} {\p}) \neq \emptyset$
for both $i=1,2$, so that $\eset_0, \eset_1 \neq \emptyev$.

Again by \ref{active}, we know that $\p$ is active in both $\eset_i$,
since $\min (\Proj{\eset_1} \p)$ and $\min (\Proj{\eset_2} \p)$ are
output uniform, while condition \ref{passive} implies that all
$\q\in \Pi\setminus \p$ are passive, since
$\min (\Proj{\eset_i} {\q})$ are input uniform. We conclude that
$\wb[\eset_0][\eset_1]$ and hence that
$\Sem{\aG_1 + \aG_2} \neq \bot$.
Now that
$\widehat{\minP_1 \cup \minP_2}(\rr) = \min (\Proj{\Sem{\aG_1 +
	 \aG_2}}{\rr})$ and
$\widehat{\maxP_1 \cup \maxP_2}(\rr) = \max (\Proj{\Sem{\aG_1 +
	 \aG_2}}{\rr})$ for all $\rr \in \PartSet(\aG_1 + \aG_2)$ follows
by induction.

\begin{example}\label{exm:typingChoice}
{\em
By rule $\textsc{t-ch}$ we can type e.g.
\[
\prooftree
	\prooftree
		\minP_1  = \maxP_1 = \{\aout[c][s][][\reqmsg], \ain[c][s][][\reqmsg]\}
	\justifies
		\Set{\ptp[C],\ptp[S]} \der \gint[][c][req][s] : \typetp{\minP_1}{ \maxP_1}
	\endprooftree
	\quad
	\prooftree
		\minP_6  = \maxP_6 = \{\aout[c][s][][\donemsg], \ain[c][s][][\donemsg]\}
	\justifies
		\Set{\ptp[C],\ptp[S]} \der \gint[][c][done][s]  : \typetp{\minP_6}{ \maxP_6}
		\endprooftree
		\qquad
		\compch {\minP_1} {\minP_6}{ \Set{\ptp[C],\ptp[S]} }
\justifies
	\Set{\ptp[C],\ptp[S]} \der \gcho[][{\gint[][c][req][s]}][{\gint[][c][done][s]}]  : \typetp{\minP_1 \cup \minP_2}{ \maxP_1 \cup \maxP_2}
\endprooftree
\]
in fact, $\compch {\minP_1} {\minP_6}{ \Set{\ptp[C],\ptp[S]} }$ holds, being $\ptp[C]$ the unique participant in
$\Set{\ptp[C],\ptp[S]}$ such that $\widehat{\minP_1}(\ptp[C])$ and
$\widehat{\minP_6}(\ptp[C])$ are output uniform, and the remaining
$\ptp[S]$ is such that $\widehat{\minP_1}(\ptp[S])$ and
$\widehat{\minP_6}(\ptp[S])$ are input uniform. However none of the
following are typable; recall that
$\Set{\ptp[C],\ptp[S]} \der \gint[][c][req][s] : \typetp{\minP_1}{
  \maxP_1}$,
$\Set{\ptp[C],\ptp[S]} \der \gint[][s][done][c] : \typetp{\minP_2}{
  \maxP_2}$ and
$\Set{\ptp[C],\ptp[B]} \der \gint[][c][md][b] : \typetp{\minP_3}{
  \maxP_3}$. Then:
\begin{description}
\item $\aG_1 \equiv \gcho[][{\gint[][c][req][s]}][{\gint[][c][req][s]}]$: this is because  $\widehat{\minP_1}(\ptp[C])$ cannot be disjoint from itself;
\item $\aG_2 \equiv \gcho[][{\gint[][c][req][s]}][{\gint[][s][done][c]}]$: in this case we have that neither 
	$\widehat{\minP_1}(\ptp[C]) \cup \widehat{\minP_2}(\ptp[C]) = \Set{ \aout[c][s][][\reqmsg], \ain[s][c][][\donemsg] }$ nor
	$\widehat{\minP_1}(\ptp[S]) \cup \widehat{\minP_2}(\ptp[S]) = \Set{ \ain[c][s][][\reqmsg], \aout[s][c][][\donemsg] }$ are output uniform;
\item $\aG_3 \equiv \gcho[][{\gint[][c][req][s]}][{\gint[][c][md][b]}]$: because $\Set{\ptp[C],\ptp[S]} \neq \Set{\ptp[C],\ptp[B]}$.
\end{description}

A more complex case is the following (continuing example \ref{exm:typingSeq}):
\[
\prooftree
	\prooftree
		\minP_1  = \maxP_1 = \{\aout[c][s][][\mdmsg], \ain[c][s][][\mdmsg]\}
	\justifies
		\Set{\ptp[C],\ptp[S]} \der \gint[][c][md][s] : \typetp{\minP_1}{ \maxP_1}
	\endprooftree
	\qquad
	\Set{\ptp[C],\ptp[S]} \der \gseq[][{\gint[][c][req][s]}][{\gint[][s][done][c]}]  : \typetp{\minP_2}{ \maxP_3}
\justifies
	\Set{\ptp[C],\ptp[S]} \der \gcho[][{\gint[][c][md][s]}][{\gseq[][{\gint[][c][req][s]}][{\gint[][s][done][c]}]}]
	: \typetp{\minP_1 \cup \minP_2}{\maxP_1 \cup \maxP_3}
\endprooftree
\]
because $\ptp[C]$ is the unique participant in $\Set{\ptp[C],\ptp[S]}$ such that $\widehat{\minP_1}(\ptp[C])$ and $\widehat{\minP_2}(\ptp[C])$
are output uniform, and the remaining $\ptp[S]$ is such that $\widehat{\minP_1}(\ptp[S])$ and $\widehat{\minP_2}(\ptp[S])$ are input
uniform, namely condition $\compch {\minP_1} {\minP_2}{\Set{\ptp[C],\ptp[S]}}$ is satisfied.

On the other hand the choreography $\gcho[][{\gseq[][{\gint[][c][md][b]}][{\gint[][b][md][s]}]}][{\gseq[][{\gint[][c][req][s]}][{\gint[][s][done][c]}]}]$
is not typeable because knowing from example \ref{exm:typingSeq} that
\[
\Set{\ptp[B],\ptp[C],\ptp[S]} \der\gseq[][{\gint[][c][md][b]}][{\gint[][b][md][s]}] : \typetp{\minP_6}{ \maxP_6}
\]
we have that $\Set{\ptp[B],\ptp[C],\ptp[S]} \neq \Set{\ptp[C],\ptp[S]}$, so that rule $\textsc{t-ch}$ doesn't apply.

}\end{example}

In summary we have proved the following result.
\begin{theorem}[Soundness]\label{thr:soundness}
  If $~\tj \CxtPi {\aG} {\minP} {\maxP}$ is derivable then
  $\Sem{\aG} \neq \bot$, $\Pi = \PartSet(\aG)$,
  and
  \[
	 \widehat{\minP}(\p) = \min (\Proj{\Sem{\aG}}{\p})
	 \qqand
	 \widehat{\maxP}(\p) = \max (\Proj{\Sem{\aG}}{\p})
  \]
   holds for all $\p \in \Pi$.
\end{theorem}

\begin{corollary}\label{cor:soundness}
  If $\aG$ is a typable g-choreography then $\aG$ is
  well-formed.
  Moreover a typable $\aG$ has exactly one typing
  $\Pi \der \aG : \typetp \minP \maxP$.
\end{corollary}

Our typing system is not complete.
For instance, the g-choreography
$\gpar[][{\gint[]}][{\gint[][a][@][c]}]$ is well-formed but it cannot
be typed because the rule \textsc{t-par} cannot be applied since \p\
occurs on both sides of the parallel.
In fact, this is the only obstacle to attain completeness and could be
removed by tracing in the types not only the minimal and maximal
communications, but also the communications of threads.
\begin{lemma}\label{lem:partyping}
  If $\tj \CxtPi {\aG} {\minP} {\maxP}$ is derivable then, for all $x
  \in \MaxC{\Sem{\aG}}$ and $\p \in \PartSet$, either
  $\max(\Proj{x}{\p})$ is empty or it is a singleton.
\end{lemma}
\begin{proof}
  By induction on the derivation of $\tj \CxtPi \aG \minP \maxP$
  and case analysis of the last typing rule applied observing that
  rule \textsc{t-par} requires to partition the context in two
  disjoint sub-contexts.
\end{proof}
We remark that such a completeness result would be basically due to the
strictness of the conditions of \cref{def:wb}.
In fact, more general notions of well-branchedness would break
the completeness theorem.
For instance, we can weaken the conditions of \cref{def:wb}
as follows.
\begin{itemize}
\item The projections on the event structures of the two branches may
  either be disjoint inputs (as per the current determined choice
  condition) or be isomorphic
\item there is a unique selector (as currently required in
  \cref{def:wb})
  and any other participant whose minimal actions are output
  have isomorphic projections on the two branches.
\end{itemize}
With this change, the g-choreography $\aG = \gcho[][\aG_1][\aG_2]$
where $\aG_1 = \gseq[][{\gint[]}][{\gint[][c][x][b]}]$ and
$\aG_2 = \gseq[][{\gint[][@][n]}][{\gint[][c][x][b]}]$ would become
well formed.
However, our system cannot type $\aG$ since both \p\ and \p[c] are
selectors in the choice of $\aG$.
Notice that $\aG_1$ violates the unique selector condition of
\cref{def:wb}, while it does not violate the more general conditions
above since \p[c] behaves the same in $\aG_1$ and $\aG_2$.


\section{Refinement}\label{sec:refinement}

\newcommand{\RTint}{\textsc{\footnotesize r-int}}
\newcommand{\refines}{\;\textrm{ref}\;}

To the grammar of Definition \ref{def:refgg} we add a new construct
that we dub {\em refinable action}:
\begin{align*}
  \aG \bnfdef & \cdots \mid \refgint[@][{\msg_1} \ldots {\msg_n}][{\q_1}, \ldots, {\q_n}]
  & \text{refinable action}
\end{align*}
where $\vec \msg = \msg_1, \ldots, \msg_n$ and
$\vec \q = \q_1, \ldots, \q_n$ are non-empty tuples of the same lenght
of messages and participants such that the participants in
$\vec{\ptp[B]}$ are pairwise distinct.
Call {\em refinable} a g-choreography generated by the so extended
grammar; a g-choreography is {\em ground} or {\em non-refinable} if it
is derivable only with the productions of the grammar
in~\cref{def:refgg}.

\begin{definition}[Refines relation]\label{def:refinementRel}
  A ground g-choreography $\aG$ {\em refines}
  $\refgint[@][\vec{\msg}][\vec{B}]$, written {\em $\aG \refines
	 \refgint[@][\vec{\msg}][\vec{B}]$}, if
  \begin{enumerate}
  \item \label{def:refinementRel-i}
	 $\Sem{\aG} = \eset\neq \bot$;
  \item \label{def:refinementRel-ii}
	 $\esbj[{\min(\eset)}] = \Set{\p}$, by which we say that
	 $\p$ is the (unique) {\em initiator} of $\aG$;
  \item \label{def:refinementRel-iii} letting $\vec \msg = \msg_1,
	 \ldots, \msg_n$ and $\vec \q = \q_1, \ldots, \q_n$, for all $x \in
	 \MaxC{\eset}$ and $1 \leq h \leq n$ there exists $\ptp[C] \in
	 \PartSet(\aG)$ such that $ \ain[C][B_h][][{\msg[m_h]}][] \in \max
	 (\Proj{x}{\q_h})$.
  \end{enumerate}
\end{definition}
In words, $\aG$ refines $\refgint[@][\vec{\msg}][\vec{B}]$ if $\aG$ is
meaningful, that is well-formed, with a unique participant $\p$
initiating the interaction by some (necessarily distinct) output
actions, and such that in all branches, namely maximal configurations
$x$ of $\Sem{\aG}$,
each $\q_h$ eventually inputs $\msg_h$.

\begin{example}\label{exm:refinementRel}\rm
  The following
  \[
	 \gint[] \refines \refgint,
	 \quad
	 \gcho[][\gint][{\gseq[][{\gint[][a][n][b]}][\gint]}] \refines \refgint,
	 \qand
	 \gpar[][\gint][{\gint[][c][n][b]}] \refines \refgint
  \]
  are examples of refinement relations.
  \finex
\end{example}

Our next step is to devise sufficient conditions for substituting the
refinement action $\refgint[@][\vec{\msg}][\vec{B}]$ by some of its
ground refinements $\aG'$ in a context of the shape
$\aG[\refgint[@][\vec{\msg}][\vec{B}]]$ while ensuring
well-formednes of the resulting g-choreography $\aG[\aG']$. In view of
the previous section, an eligible tool is the typing system.
Hereafter, let $\aG'[\cdot]$ be a ground g-choreography with a hole
$[\cdot]$, namely a placeholder such that, if replaced by a ground
$\aG$ the resulting $\aG'[\aG]$ is a ground g-chreography.

In fact, observe that, barred for the axioms, the shape of the typing
judgement $\tj \CxtPi {\aG} {\minP} {\maxP}$ in the conclusion of each
rule only depends on contexts and types in the premises. Combining
this remark with Theorem \ref{thr:soundness}, we get the following
corollary.

\begin{corollary}\label{cor:compositionality}
  Suppose that
  $\Pi' \vdash \aG'[\cdot]: \langle \minP', \maxP' \rangle$
  is derivable in the type system extended by the axiom $\Pi \vdash [\cdot]: \langle \minP, \maxP \rangle$.
  Then for all $\aG$ such that $\Pi \vdash \aG: \langle \minP, \maxP \rangle$ is derivable, 
  the judgment $\Pi' \vdash {\aG'[\aG]} : \langle {\minP'}, {\maxP'} \rangle$ is derivable, so that $\Sem{\aG'[\aG]} \neq \bot$.
\end{corollary}

To put this corollary to use, we have to define an axiom schema for deducing $\tj  \CxtPi { \refgint[@][\vec{\msg}][\vec{B}] } {\minP} {\maxP}$
such that if $\tj  \CxtPi {\aG} {\minP} {\maxP}$ is derivable for some ground $\aG$, then $\aG$ refines $ \refgint[@][\vec{\msg}][\vec{B}]$.
To do that we need a preliminary lemma.

\begin{lemma}\label{lem:min-max}
  Let $\Sem{\aG} \neq \emptyev$ be ground and well-formed. If $x \in \MaxC{\Sem{\aG}}$ then
  \begin{enumerate}
  \item \label{lem:min-max-i} $\emptyset \neq \min (x) \subseteq \lset^!$ and $\emptyset \neq \max (x) \subseteq \lset^?$;
  \item \label{lem:min-max-ii} $\esbj[x] =  \PartSet(\aG)$.
  \end{enumerate}
\end{lemma}

\begin{proof}
  By induction over $\aG$.

  If $\aG \equiv \gint$ then $\Sem{\aG} = \Set{\aout < \ain}$ and the (\ref{lem:min-max-i})-(\ref{lem:min-max-ii}) are immediately verified.

  If $\aG \equiv \gpar[][\aG_1][\aG_2]$ we have that
  $\Sem{\aG} = \Sem{\aG_1} \otimes \Sem{\aG_2}$ by well-formednes. If
  either of the $\Sem{\aG_i}$ is $\emptyev$ the thesis follows
  immediately by induction, since the other one has to be
  $\neq \emptyev$. Suppose that $\Sem{\aG_i} \neq \emptyev$ for both
  $i=1,2$. By definition of $\otimes$ we have that for some non empty
  $x_i \in \MaxC{\Sem{\aG_i}}$ it is $x = x_1 \cup x_2$, from which
  part (\ref{lem:min-max-i}) of the lemma follows by induction.
  Similarly if $y \in \MaxC{\Sem{\aG}}$ then $y = y_1 \cup y_2$ for
  $y_i \in \MaxC{\Sem{\aG_i}}$, hence by induction
  $\esbj[{x_i}] = \esbj[{y_i}]$ for $i = 1, 2$, so that
  $\esbj[x] = \esbj[{x_1}] \cup\, \esbj[{x_2}] = \esbj[{y_1}] \cup
  \esbj[{y_2}] = \esbj[y]$ and we conclude that (\ref{lem:min-max-ii})
  holds.

  If $\aG \equiv \gseq[][\aG_1][\aG_2]$ then by definition of
  $\Sem{\aG} = \rseq [\Sem{\aG_1}] [\Sem{\aG_2}]$ any
  $x \in \MaxC{\Sem{\aG}}$ either $x \in \MaxC{\Sem{\aG_1}}$, or
  $x \in \MaxC{\Sem{\aG_2}}$ or there exist
  $x_1 \in \MaxC{\Sem{\aG_1}}$ and $x_2 \in \MaxC{\Sem{\aG_2}}$ with
  $x = x_1 \cup x_2$ and $\min(x) = \min(x_1)$, $\max(x) = \max(x_2)$.
  Now (\ref{lem:min-max-i}) follows by induction. To see
  (\ref{lem:min-max-ii}) suppose that $y \in \MaxC{\Sem{\aG}}$ and,
  toward a contradiction, assume that $\esbj[x] \neq \esbj[y]$: this
  is only possible if say
  $x \in \MaxC{\Sem{\aG_1}}\setminus\MaxC{\Sem{\aG_2}}$ and
  $y \in \MaxC{\Sem{\aG_2}}\setminus\MaxC{\Sem{\aG_1}}$, since
  otherwise if e.g. $x = x_1 \cup x_2$ as above then
  $\esbj[y] = \esbj[x_2]$ by induction, but then by definition of
  $ \rseq [\Sem{\aG_1}] [\Sem{\aG_2}]$ there is a pair of events
  $\ae_1 \in x_1$ and $\ae_2 \in x_2$ with
  $\esbj[{\lambda_1 (\ae_1)}] = \esbj[{\lambda_2(\ae_2)}]$ so that by
  construction $\ae_1 \leq_{\Sem{\aG}} \ae_2$.  This implies that
  there exist $\ae_3 \in y$ s.t.
  $\esbj[{\lambda_2 (\ae_3)}] = \esbj[{\lambda_2 (\ae_2)}]$ so that
  $\ae_1 \leq_{\Sem{\aG}} \ae_3$ contradicting the maximality of
  $y \in \MaxC{\Sem{\aG}}$.

  But if $x \in \MaxC{\Sem{\aG_1}}\setminus\MaxC{\Sem{\aG_2}}$ and
  $y \in \MaxC{\Sem{\aG_2}}\setminus\MaxC{\Sem{\aG_1}}$ then
  $x \neq y$ and $x \cup y \in \MaxC{\Sem{\aG}}$ contradicting the
  hypothesis that $x, y \in \MaxC{\Sem{\aG}}$.

  If $\aG \equiv \gcho[][\aG_1][\aG_2]$ and it is well-formed then
  $\PartSet(\aG_1) = \PartSet(\aG_2) = \PartSet(\aG)$ and there exists
  a unique active $\p \in \PartSet(\aG)$ in
  $\Sem{\aG} = \Sem{\aG_1} + \Sem{\aG_2}$ that is the subject of all
  events in $\min(x)$ for any $x \in \MaxC{\Sem{\aG}}$. This also
  implies that $\Sem{\aG} \neq \emptyev$, so that at least one of the
  two $\Sem{\aG_i}$ is such. By induction we have immediately
  (\ref{lem:min-max-i}). To prove (\ref{lem:min-max-ii}) let
  $y \in \MaxC{\Sem{\aG}}$. By definition of
  $\Sem{\aG_1} + \Sem{\aG_2}$ we have that
  $\MaxC{\Sem{\aG}} = \MaxC{\Sem{\aG_1}} \cup \MaxC{\Sem{\aG_2}}$, so
  that either $x$ and $y$ belong to the same $\MaxC{\Sem{\aG_i}}$,
  then $\esbj[x] = \esbj[y]$ by induction, or say
  $x \in \MaxC{\Sem{\aG_1}}$ and $y \in \MaxC{\Sem{\aG_2}}$: then by
  induction $\esbj[x] = \PartSet(\aG_1) = \PartSet(\aG_2) = \esbj[y]$
  and we are done.
\end{proof}

\begin{lemma}\label{lem:t-refConditions}
  Let $\refgint[@][\vec{\msg}][\vec{\q}]$ be a refinable action. If $\CxtPi \subseteq \PartSet$ and
  $\minP, \maxP \subseteq \lset$ are such that
  \begin{enumerate}
  \item $\esbj[{\minP}] =  \esbj[{\maxP}] = \CxtPi$,
  \item $\esbj[{(\minP \cap  \lset^! )}] = \Set{\p}$, and
  \item assuming $\vec \q = \q_1, \ldots, \q_n$ and $\vec{\msg} =
	 \msg_1, \ldots, \msg_n$, for all $1 \leq h \leq n$ there exists $\ptp[C]$ such that
	 $\widehat{\maxP}(\ptp[B_h]) = \Set{\ain[C][B_h][][{\msg[m_h]}][]}$
  \end{enumerate}
  then $\tj \CxtPi \aG \minP \maxP$ implies {\em $\aG \refines  \refgint[@][\vec{\msg}][\vec{B}]$}.
\end{lemma}
\begin{proof}
  By Theorem \ref{thr:soundness}, we know that  $\tj \CxtPi \aG \minP \maxP$ implies $\Sem{\aG} \neq \bot$; on the other hand the
  hypotheses imply that none among $\CxtPi, \minP$ and $\maxP$ is empty, hence $\Sem{\aG} \neq \emptyev$. Recall that, by the same theorem, 
  $\widehat{\minP}(\ptp[c]) = \min (\Proj{\Sem{\aG}}{\ptp[c]})$ and $\widehat{\maxP}(\ptp[c]) = \max (\Proj{\Sem{\aG}}{\ptp[c]})$
  for all $\ptp[c] \in \CxtPi  = \PartSet(\aG)$.

  Let $x \in \MaxC{\Sem{\aG}}$. We have $\emptyset \neq \min(x) \subseteq \minP \cap \lset^!$ by the above and by 
  Lemma \ref{lem:min-max}.\ref{lem:min-max-i}, hence $\esbj[{(\minP \cap  \lset^! )}] = \Set{\p}$ implies
  $\esbj[{\min(x)}] = \Set{\p}$, and therefore $\p$ is the initiator of $\aG$.

  On the other hand, by the hypothesis that for all $\ptp[B_h] \in  \vec{\ptp[B]}$ there exists $\ptp[c]$ s.t.
  $\widehat{\maxP}(\ptp[B_h]) = \Set{\ain[C][B_h][][{\msg[m_h]}][]}$ an Theorem \ref{thr:soundness} we infer that
  $\Set{\ain[C][B_h][][{\msg[m_h]}][]} = \widehat{\maxP}(\ptp[B_h] ) = \max (\Proj{\Sem{\aG}}{\ptp[B_h] })$ for all
  $\ptp[B_h] \in  \vec{\ptp[B]}$. Therefore there exists $y \in \MaxC{\Sem{\aG}}$ such that $\ain[C][B_h][][{\msg[m_h]}][] \in y \cap \maxP$;
  by Lemma \ref{lem:min-max}.\ref{lem:min-max-ii}, $\esbj[x] = \esbj[y]$, hence $\ptp[B_h] \in \esbj[x]$ and so
  $\widehat{x}({\ptp[B_h] }) \neq \emptyset$.
  This and the the typability of $\aG$ imply that $\widehat{x}({\ptp[B_h] })$ is a singleton by Lemma \ref{lem:partyping}:
  from this and the fact that  $\max(\Proj{x}{\ptp[B_h] }) \subseteq \maxP$ it follows 
  that $\max(\Proj{x}{\ptp[B_h] }) = \widehat{\maxP}(\ptp[B_h] ) = \Set{\ain[C][B_h][][{\msg[m_h]}][]}$.
\end{proof}

In view of Lemma \ref{lem:t-refConditions} it is sound to extend the type system to refinable choreographies by adding to the rules
in Figure \ref{fig:typing} the following axiom schema:
\[
  \arule{
	 \esbj[\minP]=\esbj[\maxP]= \CxtPi \quad
	 \esbj[{(\minP \cap  \lset^! )}] = \Set{\p} \quad 
	 \forall h\, \exists \ptp[C]\in\CxtPi. \;\widehat{\maxP}(\ptp[B_h]) = \Set{\ain[C][B_h][][{\msg[m_h]}][]}
  }{
	 \tj  \CxtPi {\refgint[@][{\msg_1} \ldots {\msg_n}][{\q_1}, \ldots, {\q_n}] } \minP \maxP
  }{
	 t-ref
  }
\]

\begin{remark}\label{rem:refinable}\it
  Given any refinable action
  $\refgint[@][\vec{\msg}][\vec{B}] \equiv \refgint[@][{\msg_1} \ldots
  {\msg_n}][{\q_1}, \ldots, {\q_n}]$ there is a typing context
  $\CxtPi_0, \minP_0, \maxP_0$ such that
  $\tj {\CxtPi_0} { \refgint[@][\vec{\msg}][\vec{B}] } {\minP_0}
  {\maxP_0}$ is an instance of rule \textsc{t-ref}. Indeed taking
  $\CxtPi_0 = \Set{\p, \ptp[B_1], \ldots , \ptp[B_n]}$ and
  $ \minP_0 = \maxP_0 = \Set{\aout[a][b_1][][\msg_1],
	 \ain[a][b_1][][\msg_1], \ldots, \aout[a][b_n][][\msg_n],
	 \ain[a][b_n][][\msg_n]} $ it is easy to check that the conditions
  of rule \textsc{t-ref} are satisfied. On the contrary the same
  conditions do not ensure that there is a $\aG$ that is typable in
  the given context: let
  $\CxtPi' = \Set{\ptp[a], \ptp[b], \ptp[c], \ptp[d]}$ and
  $\minP' = \maxP' = \Set{\aout[a][b][][m], \ain[a][b][][m],
	 \ain[c][d][][m]} $, then
  $\tj {\CxtPi'} {\refgint[a][m][b]}{\minP'}{\maxP'}$, but there
  exists no ground $\aG$ such that
  $\tj {\CxtPi'} {\aG}{\minP'}{\maxP'}$.
\end{remark}

\newcommand{\TRef}{ \textsc{\footnotesize t-ref} }

\begin{example}\label{exm:refinements}\rm
  The refinements in the above proof are the simplest, but less
  interesting ones; to see more significant examples let us first
  generalize Corollary \ref{cor:compositionality} to the case of
  contexts $\aG[\cdot]_1 \ldots [\cdot]_n$, with $n$ distinct holes.
  Suppose that
  $\tj {\CxtPi_i} { \refgint[\p_i][\vec{\msg_i}][\vec{B}_i] }
  {\minP_i} {\maxP_i}$ are the instances of \textsc{t-ref} that have
  been used in deriving:
  \[
	 \tj \CxtPi {\aG[\refgint[\p_1][\vec{\msg_1}][\vec{B}_1] ]_1 \ldots [\refgint[\p_n][\vec{\msg_n}][\vec{B}_n] ]_n}{\minP}{\maxP}
  \]
  If $\tj  {\CxtPi_i} { \aG_i } {\minP_i} {\maxP_i}$ are derivable for ground $\aG_i$ then
  $
  \tj \CxtPi {\aG[\aG_1 ]_1 \ldots [\aG_n ]_n}{\minP}{\maxP}
  $
  is derivable, and hence $\aG[\aG_1 ]_1 \ldots [\aG_n ]_n$ is well-formed by Theorem \ref{thr:soundness}.

  Resuming from the Introduction and adapting from Example \ref{exm:typingChoice} we have:
  {\small
	 \begin{equation}\label{eq:derivationOfRefinableChor}
		\prooftree
		\prooftree
		\justifies
		\CxtPi \der {\refgint[c][md][s]}: \typetp{\minP_1}{ \maxP_1}
		\using \TRef
		\endprooftree
		\qquad
		\prooftree
		\prooftree
		\justifies
		\CxtPi \der \refgint[c][req][s] : \typetp{\minP_2}{ \maxP_2}
		\using \TRef
		\endprooftree
		\quad
		\prooftree
		\justifies
		\CxtPi \der \refgint[s][done][c]  : \typetp{\minP_3}{ \maxP_3}
		\using \TRef
		\endprooftree
		\justifies
		\CxtPi \der \gseq[][{\refgint[c][req][s]}][{\refgint[s][done][c]}]  : \typetp{\minP_2}{ \maxP_3}
		\endprooftree
		\justifies
		\CxtPi \der \gcho[][{\refgint[c][md][s]}][{\gseq[][{\refgint[c][req][s]}][{\refgint[s][done][c]}]}]
		: \typetp{\minP_1 \cup \minP_2}{\maxP_1 \cup \maxP_3}
		\using \textsc{\footnotesize t-ch}
		\endprooftree
	 \end{equation}
  }
  where $\CxtPi = \Set{\ptp[C],\ptp[S]}$,
  $\minP_1 = \maxP_1 = \{\aout[c][s][][\mdmsg],
  \ain[c][s][][\mdmsg]\}$,
  $\minP_2 = \maxP_2 = \{\aout[c][s][][\reqmsg],
  \ain[c][s][][\reqmsg]\}$ and $\minP_3 = \maxP_3 =$ \\
  $\{\aout[s][c][][\donemsg], \ain[s][c][][\donemsg]\}$.  From Example
  \ref{exm:typingChoice} we know that we can derive
  $\CxtPi \der \gint[][c][md][s] : \typetp{\minP_1}{ \maxP_1}$,
  $\CxtPi \der \gint[][c][req][s] : \typetp{\minP_2}{ \maxP_2}$ and
  $\CxtPi : \typetp{\minP_3}{ \maxP_3}$; hence we conclude that
  $\CxtPi \der
  \gcho[][{\gint[][c][md][s]}][{\gseq[][{\gint[][c][req][s]}][{\gint[][s][done][c]}]}]
  : \typetp{\minP_1 \cup \minP_2}{\maxP_1 \cup \maxP_3}$ is derivable,
  and so semantically well-formed.
  
  Consider now the g-choreography
  $\gseq[][{\gint[][c][md][b]}][{\gint[][b][md][s]}]$, which can be
  checked to refine $\refgint[c][md][s]$: refining the latter with
  the former in the context
  $\gcho[][{\refgint[c][md][s]}][{\gseq[][{\gint[][c][req][s]}][{\gint[][s][done][c]}]}]$
  spoils well-formed, and indeed
  $\CxtPi \not\der
  {\gseq[][{\gint[][c][md][b]}][{\gint[][b][md][s]}]}:
  \typetp{\minP_1}{ \maxP_1}$.  However it suffices to take
  $\CxtPi' = \CxtPi \cup \Set{\ptp[B]} = \Set{\ptp[C],\ptp[S],
	 \ptp[B]}$ and $\minP'_1 = \minP_1 \cup \Set{\ain[c][b][][md]}$,
  $\maxP'_1 = \maxP_1 \cup \Set{\aout[b][s][][md]}$ to have that
  $\CxtPi' \der {\gseq[][{\gint[][c][md][b]}][{\gint[][b][md][s]}]}:
  \typetp{\minP'_1}{ \maxP'_1}$ is derivable and
  $\CxtPi' \der {\refgint[c][md][s]}: \typetp{\minP'_1}{ \maxP'_1}$
  is an instance of \textsc{t-ref}. Incidentally we note that
  $\CxtPi', \minP'_1, \maxP'_1$ are computable from
  $\gseq[][{\gint[][c][md][b]}][{\gint[][b][md][s]}]$.

  Still we cannot freely replace $\CxtPi' \der {\refgint[c][md][s]}: \typetp{\minP'_1}{ \maxP'_1}$ in the derivation
  (\ref{eq:derivationOfRefinableChor}), because $\CxtPi' \neq \CxtPi$ where the difference
  is $\ptp[b]$, making rule \textsc{t-ch} inapplicable.
  However, by computing the typing context of $\gseq[][{\gint[][c][x][b]}][{\gint[][b][req][s]}]$
  we obtain $\CxtPi' \der \gseq[][{\gint[][c][x][b]}][{\gint[][b][req][s]}] : \typetp{\minP'_2}{\maxP'_2}$ where
  $\minP'_2 = \Set{\aout[c][b][][x], \ain[c][b][][x], \aout[b][s][][rep]}$ and
  $\maxP'_2 = \Set{\ain[c][b][][x], \aout[b][s][][rep], \ain[b][s][][rep]}$ we can
  easily check that  $\CxtPi' \der \refgint[c][req][s] : \typetp{\minP'_2}{\maxP'_2}$ is an instance of
  \textsc{t-ref}, so that $\gseq[][{\gint[][c][x][b]}][{\gint[][b][req][s]}]  \refines \refgint[c][req][s]$ by 
  Lemma \ref{lem:t-refConditions}. In this case we do not need to further modify the derivation (\ref{eq:derivationOfRefinableChor})
  to deduce $\CxtPi \der \gseq[][{\refgint[c][req][s]}][{\refgint[s][done][c]}]  : \typetp{\minP'_2}{ \maxP'_3}$ with
  $\maxP'_3 = \maxP_3 \cup \Set{\aout[b][s][][rep]}$ and then, applying rule \textsc{t-ch} we eventually obtain a derivation of
  $\CxtPi' \der \gcho[][{\refgint[c][md][s]}][{\gseq[][{\refgint[c][req][s]}][{\refgint[s][done][c]}]}]
  : \typetp{\minP'_1 \cup \minP'_2}{\maxP'_1 \cup \maxP'_3}$,
  that is the typing context in which we can safely replace $\gseq[][{\gint[][c][md][b]}][{\gint[][b][md][s]}]$ to
  $\refgint[c][md][s]$, $\gseq[][{\gint[][c][x][b]}][{\gint[][b][req][s]}]$ to $\refgint[c][req][s]$ and $\gint[][s][done][c]$ to
  $\refgint[s][done][c]$.
  \finex
\end{example}


\section{Conclusions \& Related Work}\label{sec:conc}
We proposed a framework for refining the global views of choreographies.
In the context of concurrent and distributed systems, refinement
methods have received great attention in the 80-90's.
Action refinement has been studied in different settings by 
adding  refinement combinators to, e.g.,  process algebras~\cite{aceto1994adding}, 
labelled event structures~\cite{van1989equivalence} and 
causal trees~\cite{darondeau1993refinement}.
The cornerstone in this line of work is that actions, considered
atomic at a given level of abstraction, are refined into processes
or computations, which are non-atomic at a lower level of abstraction.
For instance, a labelled event structure can be refined into another one 
by substituting all events that have a particular label by an
event structure~\cite{van1989equivalence}. Analogously, a term of
 a process algebra can be refined into another one by replacing all occurrences of 
a particular action by a term.
A straightforward application of this approach to global choreographies would 
suggest to consider a standard language for choreographies, e.g., the one formalised 
in~\cite{gt18} and
reproduced  in \cref{def:refgg}, and then provide a substitution mechanism for 
its atomic actions, which in this setting would be interactions $\gint$.
Despite being technically possible, we opted for a generalised version
 $\refgint[@][{\msg_1} \ldots {\msg_n}][{\q_1}, \ldots, {\q_n}]$ that allows 
for the explicit definition of several participants involved in a complex, underspecified 
interaction. In this way, we provide more flexibility for abstraction. 
Suppose we would like to state that a participant $\p$ is intended to communicate 
two messages $\msg_1$ and $\msg_2$ respectively to $\q_1$ and $\q_2$ but $\p$ is 
uninterested on the way in which those messages are delivered, e.g., $\p$ is equally satisfied if 
$\msg_1$ arrives first to $\q_1$ and then $\q_1$ sends $\msg_2$ to $\q_2$, or 
if the flow involves first $\q_2$ and then $\q_1$, or else if a (non-specified) 
broker sends the messages to both 
receivers. A  specification of all behaviours described above in terms of 
binary interactions leaves little space for abstraction.  This has been the main 
motivation for the introduction of refinable interactions. 
We may have opted for more general versions of refinable interactions, e.g., 
by considering multiple senders on the left hand side of the arrow. We
opted for the current  presentation because the standard well-formedness condition on global choreographies
introduces several limitations on the way in which such abstractions could be implemented, e.g., 
they could not be refined as a choices because of the standard constraint 
about single selectors. We plan to investigate suitable generalisations on the 
shape of refinable interactions. 

  A semantics of g-choreographies in terms of
  pomsets~\cite{pratt1986modeling} has been introduced
  in~\cite{gt16,gt18}.
  Pomsets can be envisaged as event structures with an empty
  conflict relation.
  This semantics captures a more general notion of well-brancheness
  than the one considered here.
  In principle, one could borrow this more general notion of
  well-formedness in our framework at the cost of increasing the
  technical execution.
Despite we associate global choreographies with event structures (as
previously done, e.g., in~\cite{castellani2019event} to give semantics
to multi-party session types~\cite{honda16jacm}), we remark that
refinement techniques developed for event structures cannot be
straightforwardly lifted to the language of global choreographies
because the semantics of interactions is given in terms of two events.
Hence, the refinement of an interaction would translate into the refinement of 
a sub-structure instead of a single event. 
This  establishes an interesting connection 
with previous work that we plan to investigate further.  
Along the same lines, the main focus on  previous work on 
refinement~\cite{aceto1994adding,van1989equivalence,darondeau1993refinement}
is concerned with  the consistency of refinement with respect to the 
 semantics of
the language, i.e., whether refinement preserves behavioural equivalences. 
Note that we have left implicit the semantics of refinable interactions, which 
is given in terms of the set of  concrete realisations that it admits. An interesting line of 
work that we envisage for future development is whether the proposed refinement preserves
equivalences.

The typing of ground g-choreography is unique
(cf. Corollary~\ref{cor:soundness}).
Actually, a simple inspection of the typing rules in \cref{fig:typing}
shows that type inference is trivial for ground g-choreography.
This is not the case for non-ground g-choreographies where a type
inference algorithm has to \quo{guess} the types of refinable actions.
We leave type inference open and we will address it in the future.

We also leave open the problem of which properties are maintained
through the refinement process, namely what can be established of
abstract programs that still holds of concrete ones. In particular, we
would like to preserve well-formedness of g-choreographies along
refinements.



\bibliographystyle{eptcs}
\bibliography{bib}

\end{document}
